\documentclass[journal]{IEEEtranTCOM}%双栏

\usepackage{amsmath, amsthm, amsfonts, amssymb, amsbsy}

\usepackage{graphicx}

\usepackage{pstricks}
\usepackage{algorithm}
\usepackage{multirow}
\usepackage{algorithmic}
\usepackage{bm} 
\usepackage{cite}
\usepackage{url} 
\usepackage{enumerate}
\usepackage{float}
\usepackage{amsmath} 
\usepackage{amssymb}
\usepackage{mathrsfs}
\usepackage{lipsum} % 调用跨栏长公式宏包
\usepackage{amsmath,epsfig,amssymb,subfigure,bm,dsfont}
\usepackage{multicol} 
\usepackage{booktabs}
\usepackage{threeparttable}
\usepackage{color}
\usepackage{xcolor}
\usepackage{soul}

\newtheorem{proposition}{Proposition}

\newtheorem{corollary}{Corollary}
\newtheorem{lemma}{Lemma}

\newcommand{\twotriangle}{\hfill $\bigtriangleup \bigtriangleup$  }
\newcommand{\eax}{\twotriangle  \end{example}}
\newcommand\bim{\begin{itemize}}
\newcommand\eim{\end{itemize}}

\begin{document}

\title{
	{An Anti-Interference AFDM System: Interference Impacts Analyses and Parameter Optimization}
}

\author{
	Peng Yuan, Zulin Wang, {\it Member, IEEE}, Tao Luo and Yuanhan Ni, {\it Member, IEEE}
	\thanks{Part of this paper has been accepted by the IEEE 2024 Global Communications Conference (GLOBECOM) Workshops \cite{yuanAFDM}. (corresponding author: Yuanhan Ni.)}
	\thanks{P. Yuan, Z. Wang, T. Luo, Y. Ni are with the School of Electronic and Information Engineering, Beihang University, Beijing 100191, China (e-mail: yuanpeng9208@buaa.edu.cn; wzulin@buaa.edu.cn; luotao22@buaa.edu.cn; yuanhanni@buaa.edu.cn).}
}

\maketitle

\begin{abstract}
	
	This paper proposes an anti-interference affine frequency division multiplexing (AFDM) system to ensure reliability and resource efficiency under malicious high-power interference originating from adversarial devices in high-mobility scenarios.
	Closed-form expressions of interferences in the discrete affine Fourier transform (DAFT) domain are derived by utilizing the stationary phase principle and the Affine Fourier transform convolution theorem, which indicates that interference impacts can be classified into stationary and non-stationary categories.
	On this basis, we reveal the analytical relationship between packet throughput and the paramerters of spread spectrum and error correction coding in our proposed anti-interference system, which enables the design of a parameter optimization algorithm that maximizes packet throughput.  
	For reception, by jointly utilizing the autocorrelation function of spreading sequence and the cyclic-shift property of AFDM input-output relation, we design a linear-complexity correlation-based DAFT domain detector (CDD) capable of achieving full diversity gain, which performs correlation-based equalization to avoid matrix inversion.
	Numerical results validate the accuracy of the derived closed-form expressions and verify that the proposed anti-interference AFDM system could achieve high packet throughput under interference in high-mobility scenarios.
	
\end{abstract}


\section{Introduction}

Next-generation mobile communication systems (such as beyond 5G and 6G) have been envisioned to show high reliability and resource efficiency in high-mobility scenarios (e.g., flying vehicles and low-earth-orbit satellites) \cite{IMT2021White,Mobility}. 
{{Meanwhile, communication systems are increasingly exposed to malicious interference emitted by external interference devices, in which adversaries intentionally transmit disruptive signals \cite{amuru2015optimal}.
Such malicious interference can impair legitimate communication and cause severe degradation in link reliability and resource utilization, and typical forms include tone interference, sweeping interference, broadband interference, and narrowband interference \cite{Jamming1}.
Consequently, it is essential to develop communication systems that show robustness in doubly selective channels caused by high-mobility while effectively resisting interference originating from adversarial interference sources \cite{AFDM_TWC,Jamming2}.}}

To cope with high-mobility scenarios, various advanced modulation schemes have been developed to achieve full diversity in doubly selective channels.
While orthogonal frequency division multiplexing (OFDM) deployed in 4G and 5G achieves high spectral efficiency \cite{3GPP}, OFDM is unable to achieve full diversity in frequency selective channels \cite{hadani2017orthogonal}. 
{{To address this limitation, orthogonal chirp-division multiplexing (OCDM) is proposed in\cite{Ouyang2016Orthogonal} based on the discrete Fresnel transform. By spreading each information symbol across the entire bandwidth, OCDM is able to exploit full frequency diversity in frequency-selective channels, thereby outperforming OFDM. However, OCDM fails to achieve full diversity in time-selective channels, which limits its performance in high-mobility scenarios \cite{Omar2016Performance}.
To further enhance diversity in doubly selective channels, orthogonal time–frequency space (OTFS) modulation is proposed in \cite{hadani2017orthogonal} by multiplexing information in the delay-Doppler (D-D) domain to combat the dynamics of doubly selective channels \cite{raviteja2019effective}. Nevertheless, OTFS suffers from excessive channel estimation overhead due to two-dimensional pilot guard structure, which limits its practical efficiency.}}

{{To deal with the drawbacks of OCDM and OTFS, affine frequency-division multiplexing (AFDM) was recently proposed based on the discrete affine Fourier transform (DAFT) \cite{AFDM_propose}. AFDM provides comparable communication performance in terms of bit error rate (BER) comparable to OTFS while requiring less channel estimation overhead \cite{yin2024diagonally,ni2025ambiguity}, making it a promising candidate for scheme in high-mobility scenarios \cite{luo2024afdm,ni2025integrated}.
Whereas, existing AFDM detectors either rely on maximum likelihood (ML) detection with prohibitive computational complexity \cite{AFDM_TWC} or minimum mean-square error (MMSE) detection with cubic complexity \cite{AFDM_DSE}, limiting their practicality. Furthermore, the reliability of AFDM in the presence of malicious interference generated by adversarial devices remains underexplored, as the interference impact in the DAFT domain have yet to be formally analyzed.}}

{{Anti-interference strategies typically integrate spread spectrum (SS), error-correction coding (ECC), and interleaving to mitigate interference and burst errors \cite{forouzan2002performance}. However, the effectiveness of anti-interference strategies depends critically on parameter optimization, particularly in hostile scenarios \cite{Jamming2}.}} In \cite{matolak2006spectrally}, the parameters are appropriately selected to optimal spectrum efficiency of mutiple carriers under interference, while the impulse-like autocorrelation function of the spreading sequence is employed to combat multipath.
In \cite{peng2019anti}, the signal-to-interference ratio (SIR) is improved under a fixed weighted cost by Q-learning.
In \cite{im2022anti}, an optimal hopping cost is derived for multi-band ad hoc networks under the constraint of a fixed communication rate.
Nevertheless, the failure to balance resource efficiency and anti-interference capability limits the overall performance of anti-interference systems.

To balance anti-interference capability and resource efficiency, packet throughput is widely used as a guiding metric for parameter optimization, as it provides a realistic indicator of both aspects \cite{zhang2023packet}.
{{In \cite{nardelli2013throughput}, the relationship between packet throughput and the ECC coding rate in additive white Gaussian noise (AWGN) channels is derived, enabling throughput-oriented optimization of the coding rate.
Nevertheless, the applicability of AWGN-based insights diminishes in doubly selective channels arising from high mobility.
Although AFDM can achieve high communication performance under interference-free high-mobility conditions, the optimal design of DAFT domain parameters under malicious interference and doubly selective fading remains unclear due to the insufficient characterization of interference impact in the DAFT domain.
This motivates a rigorous analysis of interference impact in the DAFT domain and a corresponding parameter optimization algorithm to ensure reliable AFDM performance in adversarial and high mobility scenarios.}}

In this paper, we propose an anti-interference AFDM system to ensure reliability and resource efficiency under malicious interference in high-mobility scenarios.
{{We theoretically analyze the closed-form expressions characterizing how malicious interference impact in the DAFT domain.
On this basis, we establish an explicit relationship between packet throughput and the parameters of spread spectrum and error correction coding in doubly selective channels.
To maximize packet throughput, we develop a parameter optimization algorithm that determines the optimal spreading sequence length under given ECC settings.}}
Furthermore, a linear-complexity symbol detection method is designed for the proposed anti-interference AFDM system which could achieve full diversity gain.
Numerical results validate the accuracy of the derived closed-form expressions and verify superior packet throughput of the proposed system, compared to AFDM, OTFS, and OFDM. The main contributions of this work are summarized as follows:
\begin{itemize}
	{ \item  We derive closed-form expressions of interferences in the DAFT domain that intuitively reveal interference impact.
		We observe that the intractable finite quadratic exponential summation in the DAFT fails to intuitively reveal interference impacts and thus {{hinders}} anti-interference system design.
		To address this challenge, we propose an analysis method by utilizing the AFT convolution theorem and the stationary phase principle to derive closed-form expressions that intuitively reveal the impacts.
		
		\item We establish the analytical relationship between packet throughput and the system {{parameters}}, which enables the design of a parameter optimization algorithm that maximizes packet throughput. 
		We analytically derive the relationship between packet throughput and the SS and ECC {{parameters}} under various types of interference and doubly selective channel conditions, in contrast to the AWGN channel assumption in \cite{nardelli2013throughput}. Accordingly, {{we develop an optimization algorithm that determines the optimal spreading sequence length for fixed ECC parameters, where the nonlinear objective is efficiently solved using Newton's method.}}
		
		\item We develop a linear-complexity correlation-based DAFT domain detector (CDD) for the proposed anti-interference AFDM system, which could achieve full diversity gain.
		By jointly utilizing the impulse-like autocorrelation function of {{the}} spreading sequence and the cyclic-shift property of AFDM input-output relation, CDD performs a correlation-based operation to avoid matrix inversion, thereby reducing the complexity. {{Both}} theoretical analysis and simulation results indicate that the proposed CDD achieves full diversity with linear complexity.}	
\end{itemize} 

The rest of this paper is organized as follows. In Section II, the signal model relevant to the proposed anti-interference AFDM system is stated. Theoretical analyses of typical malicious interferences in the DAFT domain are derived in Section III. In Section IV, the analytical relationship between packet throughput and the system {{parameters}} and a parameter optimization method are proposed. A linear-complexity symbol detection method is introduced in Section V. Numerical results are given in Section VI. Finally, Section VII concludes the paper.

\textit{Notations:}
$a$, $\bf{a}$ and $\bf{A}$ represent scalar, vector, and matrix, respectively.
The blackboard bold letters $\mathbb{C}$, $\mathbb{Z}$, $\mathbb{N}$, $\mathbb{E}$, and $\mathbb{V}$ denote the complex number field, the integer field, the  natural number field, the expectation operator, and the variance operator, respectively.
${\bf{I}}_N$ denotes the $N$-dimensional identity matrix.
${\left(  \cdot  \right)^{\rm{H}}}$, ${\left(  \cdot  \right)^{\rm{T}}}$, and ${\left(  \cdot  \right)^*}$ represent the Hermitian, the transpose and the conjugate operations, respectively.
$\odot$ and $\otimes$ denote the Hadamard product operator and the Kronecker product operator, respectively. ${\left\langle {\cdot} \right\rangle _N}$ denotes the modulus operation with respect to $N$. ${ \binom{n}{r}}$ is the number of combinations of $n$ items taken $r$ at a time. $\left\lfloor {\cdot} \right\rfloor $ and ${\left\lceil {\cdot} \right\rceil }$ are the floor function and the ceil function, respectively. $\delta \left(  \cdot  \right)$ and $Q\left(\cdot \right)$ denote the Dirac delta function and the right-tail function of the standard normal distribution, respectively. $U\left(a,b\right)$ denotes a uniform distribution over the interval $\left[a,b\right]$, while $\Gamma\left(\alpha,\beta\right) $ represents a Gamma distribution with shape parameter $\alpha$ and scale parameter $\beta$.

\section{Preliminaries}
To describe the proposed anti-interference AFDM system clearly, we state AFDM fundamentals, typical malicious interference model, and {{the}} performance metric in this section.
\subsection{AFDM Fundamentals}

The Affine Fourier transform (AFT) and the DAFT form the basis of AFDM \cite{AFDM_TWC}. The AFT is a four-parameter ${\left(a,b,c,d\right)}$ class of linear integral transform defined as \cite{LCT}
\begin{align}\label{eq:LCT}
\vspace*{-2pt} %留空白，可自己调整
L\left( u \right) =
\begin {cases}
{\frac{1}{{\sqrt {2\pi \left| b \right|} }}\int_{ - \infty }^\infty  {s\left( t \right){K_{a,b,c,d}}\left( {t,u} \right)dt} } \quad &{b} \ne 0\\
{\frac{1}{{\sqrt {\left| a \right|} }}s\left( {du} \right){e^{ - j\left( {\frac{{cd}}{2}{u^2}} \right)}}} \quad &{b} = 0\\
\end{cases},
\vspace*{-2pt} %留空白，可自己调整
\end{align}
where ${K_{a,b,c,d}}\left( {t,u} \right)$ is the transform kernel given by
\begin{equation}
\vspace*{-2pt} %留空白，可自己调整
{
	\label{eq:KA}
	{K_{a,b,c,d}}\left( {t,u} \right) = {e^{ - j\left( {\frac{{a{u^2} + 2ut + d{t^2}}}{{2b}}} \right)}}\text{.} }
\vspace*{-3pt} %留空白，可自己调整
\end{equation}

To derive the DAFT, $s\left(t\right)$ and $L\left(u\right)$ are sampled by the interval $\Delta t$ and $\Delta u$ as
\begin{equation}
\vspace*{-2pt} %留空白，可自己调整
{
	\label{eq:AFT_sam}
	s\left( n \right) = s\left( t \right){|_{t = n\Delta t}},S\left( m \right) = L\left( u \right){|_{u = m\Delta u}}\text{,} }
\vspace*{-3pt} %留空白，可自己调整
\end{equation}
where $m=0,\ldots,N-1,$ and $n=0,\ldots,N-1$.
To ensure the converted AFT to be reversible, the condition
\begin{equation}
\vspace*{-2pt} %留空白，可自己调整
{
	\label{eq:DAFTcondition}
	\Delta t\Delta u = \frac{{2\pi \left| b \right|}}{N} \text{,} }
\vspace*{-2pt} %留空白，可自己调整
\end{equation}
should hold. Let ${c_1} = \frac{d}{{4\pi b}}\Delta {t^2}$ and ${c_2} = \frac{a}{{4\pi b}}\Delta {u^2}$, the DAFT is defined as
\begin{equation}
\vspace*{-2pt} %留空白，可自己调整
{
	\label{eq:DAFT}
	S\left( m \right) = \frac{1}{{\sqrt N }}{e^{ - j2\pi {c_2}{m^{2}}}}\sum\limits_{n = 0}^{N - 1} {s\left( n \right){e^{ - j2\pi \left( {\frac{{mn}}{N} + {c_1}{n^2}} \right)}}} \text{.} }
\vspace*{-2pt} %留空白，可自己调整
\end{equation}
And the inverse DAFT (IDAFT) is
\begin{equation}
\vspace*{-2pt} %留空白，可自己调整
{
	\label{eq:IDAFT}
	s\left( n \right) = \frac{1}{{\sqrt N }}{e^{j2\pi {c_1}{n^2}}}\sum\limits_{m = 0}^{N - 1} {S\left( m \right){e^{j2\pi \left( {\frac{{mn}}{N} + {c_2}{m^2}} \right)}}}  \text{.} }
\vspace*{-2pt} %留空白，可自己调整
\end{equation}

\begin{figure}
	\centering
	\includegraphics[width=3.3in]{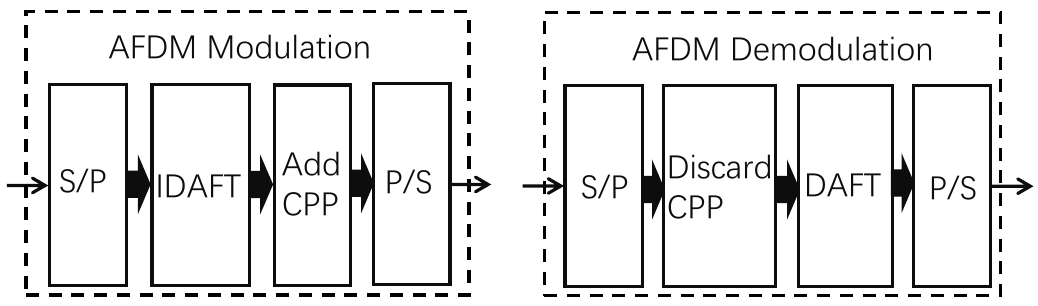}
	\vspace*{-8pt} %留空白，可自己调整
	\caption{AFDM block diagram.  
		\label{fg:AFDMDiagram}} 
	\vspace*{-15pt} %留空白，可自己调整
\end{figure}

The block diagram of AFDM is shown in Fig. \ref{fg:AFDMDiagram}. In AFDM, IDAFT is used to map data symbols into the time domain, while DAFT is performed at the receiver to obtain the effective DAFT domain channel response to the transmitted data. {{Additionally, a $N_{\rm{cp}}-$long chrip-periodic prefix (CPP), occupying the positions of the negative-index time domain samples, is used to ensure that the time delay introduced by the channel does not destroy the periodicity inherently defined by IDAFT \cite{AFDM_propose}.}}

{{Figure \ref{fg:AFDMDiagram} illustrates that DAFT is the core distinguishing feature of AFDM. To analyze the impact of interference in the DAFT domain and to enhance the waveform level anti-interference capability of the AFDM system, this paper considers single-input a single-output (SISO) model. The SISO model enables the derivation of closed form expressions for DAFT domain interference impacts and the development of a tractable parameter optimization algorithm.
Under a doubly selective channel with $L$ paths, each characterized by a delay in samples expressed as $l_i = d_i + \iota_i$, where $d_i$ is the integer delay and $\iota_i$ lies in the interval $(-{1}/{2}, {1}/{2}]$ denoting the fractional delay, together with a Doppler shift (in digital frequency) $v_i$ and a complex gain $h_i$, the AFDM input-output relation can be written as
\begin{align}\label{eq:InOut}
{\bf{y}}=\sum\limits_{i = 0}^{L - 1}  {h_i {\bf{H}}_i {\bf{x}}}+{\bf{w}},
\end{align}
where ${\bf{x}}\in {\mathbb{A}^{N \times 1}}$ denotes the vector of information symbols in the DAFT domain, ${\bf{y}}\in {\mathbb{C}^{N \times 1}}$ denotes the vector of the DAFT domain output symbols, ${{{\bf{w}}\sim \mathcal {CN}\left( {0\text{,} P_{\rm{n}}\bf{I}} \right)}}$ is additive Gaussian noise vector, and the elements of ${\bf{H}}_i$ can be given by \cite{AFDM_ISAC}
\begin{align}\label{eq:Hi}
{\bf{H}}_i[p,q]=& \frac{1}{N}e^{j\frac{2\pi}{N}(Nc_{1}l_{i}^{2}-ql_{i}+Nc_{2}(q^{2}-p^{2}))}  \nonumber \\  & \cdot \underbrace{ \sum\limits_{n = 0}^{N - 1} {{e^{j2\pi \left( {q - p - 2N{c_1}{l_i} + {k_i}} \right)n}}{e^{j2\pi {\iota _i}\varepsilon \left( {n,{l_i}} \right)}}}}_ {{\mathcal {F}}_{i}\left ({p, q}\right)},
\end{align}
where $k_{i}\triangleq N\cdot v_{i}$, 
\begin{equation}
\vspace*{-2pt} %留空白，可自己调整
{
	\label{eq:PhaseFractional}
	\varepsilon \left( {n,{l_i}} \right) = \sum\limits_{x = 0}^{2N{c_1}} {x{{I_{\boldsymbol {\mathcal {L}}_{q,x}}}({{\left\langle {n - {l_i}} \right\rangle }_N})}}  \text{,} }
\vspace*{-2pt} %留空白，可自己调整
\end{equation}	
where ${I_{\boldsymbol {\mathcal {L}}_{q,x}}}$ is the indicator function of the set ${\mathcal {L}}_{q,x}= \left[ {\left\lfloor {\frac{{N - q}}{{2N{c_1}}} + \frac{{x - 1}}{{2{c_1}}}} \right\rfloor  + 1,\left\lfloor {\frac{{N - q}}{{2N{c_1}}} + \frac{x}{{2{c_1}}}} \right\rfloor } \right]$. It is pertinent to highlight that ${{\mathcal {F}}_{i}\left ({p, q}\right)}$ in (\ref{eq:Hi}) simplifies to $\frac{{1 - {e^{ - j2\pi (p - q - {k_i} + 2N{c_1}{l_i})}}}}{{1 - {e^{ - j\frac{{2\pi }}{N}(p - q - {k_i} + 2N{c_1}{l_i})}}}}$ for zero fractional delay.}}
Since the positions of the non-zero entries of ${\bf{H}}_i$ and ${\bf{H}}_j$ ($i \ne j$) do not overlap, {{AFDM could achieve the full diversity order in doubly selective channels \cite{AFDM_TWC}.}} However, in \cite{AFDM_TWC}, maximum likelihood detection applied in AFDM is prohibitively complex to implement. Meanwhile, other detectors used in AFDM, e.g., MMSE \cite{AFDM_DSE}, may not achieve full diversity while the complexity is cubic polynomial time.

\subsection{Typical Malicious Interferences}
Typical malicious interferences, which include tone, sweeping, broadband, and narrowband interference \cite{Jamming2,Jamming3}, are briefly reviewed in this subsection.

\subsubsection{{{Tone interference}}}
{{In tone interference, one or more adversarial tones are deliberately placed across entire communication bandwidth to degrade communication performance. Tone interference includes single-tone interference and multiple-tone interference \cite{mao2016robust}. Mathematically, the signal model of tone interference is
\begin{equation}
\vspace*{-2pt} %留空白，可自己调整
{
	\label{eq:STI}
	{J_{\rm{t}}}\left( t \right) = \sqrt {\frac{{{P_{\rm{i}}}}}{{{N_{\rm{i}}}}}} \sum\limits_{k = 0}^{{N_{\rm{i}}} - 1} {{e^{j\left( {2\pi {f_{{\rm{i}},k}}t + {\theta _{{\rm{i}},k}}} \right)}}} \text{,} }
\vspace*{-2pt} %留空白，可自己调整
\end{equation}
where $P_{\rm{i}}$ is the power of interference, $N_{\rm{i}} \in {{\mathbb{N}}^{+}}$ is the number of adversarial tones, ${f_{{\rm{i}},k}}$ and ${\theta _{{\rm{i}},k}}$ are the carrier frequency and initial phase of $k-$th tone, respectively. Note that when $N_{\rm{i}}=1$, (\ref{eq:STI}) corresponds to a single-tone interference, whereas when  $N_{\rm{i}}>1$, (\ref{eq:STI}) represents a multi-tone interference. The frequency of single-tone interference is randomly distributed over the entire communication bandwidth, whereas the frequencies of multi-tone interference are allocated within the interference bandwidth according to a predefined frequency spacing.}}

\subsubsection{Sweeping interference}
{{Sweeping interference continuously changes instantaneous interference frequency over time, such as linear frequency modulation interference used by adversarial devices \cite{Jamming2}. Mathematically, the signal model of sweeping interference is}} 
\begin{equation}
\vspace*{-2pt} %留空白，可自己调整
{
	\label{eq:SWI}
	{J_{\rm{sw}}}\left( t \right) = \sqrt {{P_{\rm{i}}}} {e^{{{j\left( {2\pi {f_{\rm{i}}}t + {\theta _{\rm{i}}} + \pi {\varphi _{\rm{i}}}{{\left\langle t \right\rangle }_{{T_{\rm{i}}}}}^2} \right)}}}}\text{,} }
\vspace*{-2pt} %留空白，可自己调整
\end{equation}
where $\varphi_{\rm{i}}$ denotes frequency modulation slope, $T_{\rm{i}}$ represents the period of a single frequency sweep. {{The bandwidth covered by sweeping interference can be expressed as $B_{\rm{i}}=\varphi_{\rm{i}}T_{\rm{i}}$. To achieve maximal suppression of a target communication channel, practical sweeping interference typically scan across the entire communication bandwidth.}} 

\subsubsection{Broadband interference}
{{The broadband interference, also referred to as barrage jamming \cite{shahriar2014phy}, targets the entire channel bandwidth occupied by communication systems and effectively raises the background noise level at the receiver. Broadband interference creates a higher-noise environment that makes it more difficult for the communication system to operate. Since broadband interference generates signals that are similar to broadband background noise, it is commonly modeled as a zero-mean complex Gaussian random variable \cite{Jamming1}, i.e.,}}
\begin{equation}
\vspace*{-2pt} %留空白，可自己调整
{
	\label{eq:BBI}
	{J_{\rm{bb}}}\left( t \right) = \sqrt {{P_{\rm{i}}}} z\left( t \right) \text{,} }
\vspace*{-2pt} %留空白，可自己调整
\end{equation}
where $z\left( t \right)$ conformed to a complex centered Gaussian distribution, i.e., ${z}\left( t \right) \sim \mathcal {CN}\left( {0\text{,} 1}\right)$. 

\subsubsection{Narrowband interference}
Compared with broadband interference, narrowband interference {{concentrates interfernece power on only a portion of the frequency spectrum used by the communication system, rather than the entire band, and the occupied spectral range of interference can be adjusted \cite{Jamming1}.}} Narrowband interference can be characterized either by a white Gaussian signal passing through a filter or by a randomly modulated Phase Shift Keying (PSK) signal \cite{Jamming2}. The first signal model {{is expressed as}}
\begin{equation}
\vspace*{-2pt} %留空白，可自己调整
{
	\label{eq:NBI1}
	J_{\rm{nb}}^1\left( t \right) = \sqrt {{P_{\rm{i}}}} {e^{ - j\left( {2\pi {f_{\rm{i}}}t + {\theta _{\rm{i}}}} \right)}}\int_{ - \infty }^\infty  {h\left( {{\Omega_{\rm{i}}},\tau } \right)} z\left( {t - \tau } \right)d\tau \text{,} }
\vspace*{-2pt} %留空白，可自己调整
\end{equation}
where {{$\Omega_{\rm{i}}$ represents the subset of the communication system spectrum affected by the interference, ${h\left( {{\Omega_{\rm{i}}},t } \right)}$ denotes the impulse response of a filter associated with $\Omega_{\rm{i}}$, satisfying}}
\begin{equation}
\vspace*{-2pt} %留空白，可自己调整
{
	\label{eq:ht}
	{\int_{ - \infty }^\infty  {\left| {h\left( {{\Omega_{\rm{i}}},t} \right)} \right|} ^2}dt = 1 \text{.} }
\vspace*{-2pt} %留空白，可自己调整
\end{equation}
{{Note that interference occupied spectral range ${\Omega_{\rm{i}}}$ can be an adjacent spectral segment and it can also consist of non-adjacent spectral segments. Narrowband interference that follows the signal model (\ref{eq:NBI1}) is also referred to as partial-band jamming \cite{shahriar2014phy}.}}

The second signal model is given by
\begin{equation}
\vspace*{-2pt} %留空白，可自己调整
{
	\label{eq:NBI2}
	{J_{\rm{nb}}^2}\left( t \right) = \sqrt {{P_{\rm{i}}}} {e^{j\left( {2\pi {f_{\rm{i}}}t + {\theta _{\rm{i}}}} \right)}}\sum\limits_{p = 0}^{\infty} {{a_{\rm{i}}}\left( p \right)g\left( {\frac{{t{B_{\rm{i}}}}}{2} - p} \right)}   \text{,} }
\vspace*{-2pt} %留空白，可自己调整
\end{equation}
where ${a_{\rm{i}}}\left( p \right)$ is {{typically taken as a random PSK sequence in narrowband interference to represent an unpredictable interference, and $g\left( {t} \right)$ is the rectangular window defined by
\begin{align}\label{eq:REC}
\vspace*{-2pt} %留空白，可自己调整
g\left( t \right) =
\begin {cases}
1 \quad t \in \left[ {0,\left. 1 \right)} \right.\\
0 \quad otherwise\\
\end{cases}.
%\text{,}
\vspace*{-2pt} %留空白，可自己调整
\end{align}
The pulse scaling $g\left( {\frac{{t{B_{\rm{i}}}}}{2} - p} \right)$ implies a PSK symbol period of $T_{{\rm{i}},s}=\frac{2}{B_{\rm i}}$, which determines the effective bandwidth of the narrowband interference in model (\ref{eq:NBI2}).}}

To develop a tailored anti-interference system for AFDM, it is essential to analyze the impact of typical interferences in the DAFT domain based on interference signal models.
\subsection{Performance Metric}
Generally speaking, anti-interference strategies enhance interference resistance at the cost of reduced spectral and energy efficiency. This phenomenon reveals an inherent trade-off between communication performance and anti-interference capability. Thus, {{an}} appropriate metric for design guidance is essential to achieve a balance between resource efficiency and interference mitigation capability. 

To ensure the provision of reliable communications for mobile applications \cite{IMT2021White}, several works adopt the packet-level performance as the primary indicator \cite{zhang2023packet}. From the perspective of packet-level analysis, packet throughput quantifies the number of successfully transmitted packets over time, which is 
\begin{equation}
\vspace*{-2pt} %留空白，可自己调整
{
	\label{eq:PT_de}
	\eta  = \frac{{{P_{\rm{pc}}}}}{{T_{\rm{p}}}} \text{,} }
\vspace*{-2pt} %留空白，可自己调整
\end{equation}
where $P_{\rm{pc}}$ is the probability of successful packet transmission, and $T_{\rm{p}}$ denotes packet transmission time.

Obviously, packet throughput provides a realistic indicator of resource utilization and anti-interference capability. On one hand, under resource constraints, a high packet throughput indicates that the communication system is capable of transmitting a larger amount of effective information. Moreover, packet throughput accounts for real-world conditions such as channel impairments and coding strategies, yielding an accurate representation of resource utilization \cite{zhang2023packet}. On the other hand, as interference degrades signal quality and increases packet errors, packet throughput decreases correspondingly. A system capable of maintaining higher packet throughput under interference conditions demonstrates superior robustness. Therefore, we adopt packet throughput as a guiding metric for anti-interference design in this paper. 
\section{Impacts Analyses of Malicious Interferences in the DAFT Domain}
In this section, the impacts of typical interferences in the DAFT domain {{are}} analyzed through deriving closed-form expressions of interferences in the DAFT domain.
To address the challenge posed by finite quadratic exponential summation in analyses of tone and sweeping interference, an analysis method based on the AFT convolution theorem and the stationary phase principle is adopted. 
Moreover, the analyses of both broadband and narrowband interference are conducted based on statistical characteristics.
	
\subsection{Impact Analysis of Tone Interference}
The transformation of interference from the time domain to the DAFT domain involves the intractable finite quadratic exponential summation that fails to intuitively reveal the tone interference impact. 
To address this challenge, tone interference in the continuous AFT domain is first obtained by applying the AFT convolution theorem and the stationary phase principle.
Furthermore, to overcome the issue that the sampled result of the interference in the AFT domain can not directly reveal the impact, we derive the amplitude and phase of the sampled result.
Consequently, we obtain a closed-form expression of tone interference in the DAFT domain, which reveals the impact in the DAFT domain intuitively.

Specifically, {{tone interference after down conversion and sampling at communication receiver could be modeled as
\begin{equation}
\vspace*{-2pt} %留空白，可自己调整
{
	\label{eq:STI_Dis}
	{J_{\rm{t}}}\left( n \right) = \sqrt {\frac{{{P_{\rm{i}}}}}{{{N_{\rm{i}}}}}} \sum\limits_{k = 0}^{{N_{\rm{i}}} - 1} {{e^{j\left( {2\pi {f_{{\rm{d}},k}}n + {\theta _{{\rm{i}},k}}} \right)}}} ,n = 0,1, \ldots N - 1 \text{,} }
\vspace*{-2pt} %留空白，可自己调整
\end{equation}	
where ${f_{{\rm{d}},k}} = {{f_{{\rm{m}},k}}}/{{{f_{\rm{s}}}}}$, ${{f_{{\rm{m}},k}}}=\left({{{f_{{\rm{i}},k}}} - {f_{\rm{c}}}}\right)$, $f_{\rm{s}}$ is the sampling frequency. Accordingly, tone interference in the DAFT domain can be given by
\begin{equation}
\vspace*{-2pt} %留空白，可自己调整
{
	\label{eq:STI_DAFT1}
	J_{\rm{t}}^{\rm{A}}\left( m \right) = \sqrt {\frac{{{P_{\rm{i}}}}}{{N{N_{\rm{i}}}}}} \sum\limits_{n = 0}^N {\sum\limits_{k = 0}^{{N_{\rm{i}}} - 1} {{e^{ - j2\pi \left( {\frac{{mn}}{N} + {c_1}{n^2} + {c_2}{m^2} - {f_{d,k}}n - {\theta _{{\rm{i}},k}}} \right)}}} }   \text{,} }
\vspace*{-2pt} %留空白，可自己调整
\end{equation}}}
Due to the finite quadratic exponential summation \cite{SUM} on the right-hand side of (\ref{eq:STI_DAFT1}), it is challenging to intuitively reveal the impact of tone interference.

To intuitively reveal the tone interference impact, we first analyze tone interference in the continuous AFT domain. Applying the AFT directly to signal model (\ref{eq:STI}) yields 
\begin{align}\label{eq:STI_Lu}
\vspace*{-2pt} %留空白，可自己调整
L_{\rm{f}}^{\rm{A}}\left( u \right) {=}\sqrt{\frac{1}{{ {2\pi b } }}}\int_{ - \infty }^\infty  {f_{\rm{t}}\left(t\right)} {K_{a,b,c,d}}\left( {t,u} \right)dt 
\text{,}
\vspace*{-2pt} %留空白，可自己调整
\end{align}
where ${f_{\rm{t}}}\left( t \right) = \sqrt {{P_{\rm{i}}}/{N_{\rm{i}}}} \sum\limits_{k = 0}^{{N_{\rm{i}}} - 1} {{e^{j\left( {2\pi {f_{{\rm{m}},k}}t + {\theta _{{\rm{i}},k}}} \right)}}} g\left( {t{f_{\rm{s}}}/N} \right)$.
However, (\ref{eq:STI_Lu}) introduces significant errors to the interference in the AFT domain because it neglects the aliasing effect in the AFT domain caused by sampling. 
Therefore, the AFT should be applied to $\left[f_{\rm{t}}\left( t \right)q\left( t \right)\right]$, where $q\left( t \right) = \sum\limits_{n =  - \infty }^\infty  {\delta \left( {t - {n}/{{{f_{\rm{s}}}}}} \right)} $. Let $L_{\rm{t}}^{\rm{A}}\left( u \right)$ represent the AFT of $\left[f_{\rm{t}}\left( t \right)q\left( t \right)\right]$.
Directly solving $L_{\rm{t}}^{\rm{A}}\left( u \right)$ remains challenging due to the presence of finite quadratic exponential summation. To address this, we utilize the AFT convolution theorem to refomulate $L_{\rm{t}}^{\rm{A}}\left( u \right)$ as
\begin{equation}
\vspace*{-2pt} %留空白，可自己调整
{
	\label{eq:CT}
	L_{\rm{t}}^{\rm{A}}\left( u \right) = \frac{{{e^{ - j\frac{a}{{2b}}{u^2}}}}}{{2\pi  b }}\left\{ {\left[ {L_{\rm{f}}^{\rm{A}}\left( u \right){e^{j\frac{a}{{2b}}{u^2}}}} \right] * {F_{\rm{q}}}\left( {u/b} \right)} \right\}\text{,} }
\vspace*{-2pt} %留空白，可自己调整
\end{equation}
where ${{F_{\rm{q}}}\left( {u} \right)}$ is the Fourier transform of $q\left( t \right)$.
Due to the presence of an exponential quadratic term, $L_{\rm{f}}^{\rm{A}}\left( u \right)$ is an oscillatory integral which is difficult to evaluate. 
To calculate $L_{\rm{f}}^{\rm{A}}\left( u \right)$, we utilize the stationary phase principle \cite{POSP} based on the fact that $g\left(x\right)$ in (\ref{eq:STI_Lu}) satisfies the slowly varying condition required by the stationary phase principle.
The stationary phase principle is illustrated in Lemma 1.
\begin{lemma} 
	Let $I\left( \lambda  \right) = \int_{ - \infty }^\infty {{f\left( x \right){e^{ - j\lambda x}}{e^{j\phi \left( x \right)}}}dx}$, where $f\left(x\right)$ is slowly varying compared to the rapid oscillations of the exponential term ${e^{j\phi \left( {x} \right)}}$.
	The dominant contributions to $I\left( \lambda  \right)$ come from points where the derivative of $\phi \left( {x} \right)$ vanishes.
	Then $I\left( \lambda  \right)$ can be calculated as\cite{POSP}
	\begin{align}\label{eq:SOP}
	\vspace*{-2pt} %留空白，可自己调整
	&\ I\left( \lambda  \right) \approx \sqrt {\frac{{2\pi }}{{{\phi ^{''}}\left( {{x_k}} \right)}}} f\left( {{x_k}} \right){e^{ - j\lambda {x_k}}}{e^{j\phi \left( {{x_k}} \right)}}{e^{j\frac{\pi }{4}{\mathop{\rm sgn}} \left( {{\phi ^{''}}\left( {{x_k}} \right)} \right)}}, &
	\vspace*{-2pt} %留空白，可自己调整
	\end{align}
	where ${x_k}$ is stationary point which satisfies ${\phi ^{'}}\left( {{x_k}} \right)=0$, ${\phi ^{'}}\left( {{x}} \right)$ and ${\phi ^{''}}\left( {{x}} \right)$ are the first and second derivatives of ${\phi}\left( {{x}} \right)$, respectively.
\end{lemma}

With Lemma 1 at hand, {{and according to the analysis in \cite{chassande1999time}, when the highest-order term of the phase in the integrand is no higher than quadratic, the corresponding error of the stationary-phase principle is much smaller than $1$. This guarantees the validity of applying the stationary phase principle in our derivation. Consequently, }}the tone interference in the AFT domain can be further derived by invoking the AFT convolution theorem. 

By sampling the interference in the AFT domain and analysing the sampled results, we {{obtain}} derived closed-form expression of tone interference in the DAFT domain which is provided in Proposition 1. 

\begin{proposition} 
	When $N_{\rm{i}}=1$, corresponding to single-tone interference, the tone interference in the DAFT domain can be written as
	\begin{align}\label{eq:STI_DAFT}
	\vspace*{-2pt} %留空白，可自己调整
	J_{\rm{t}}^{\rm{A}}\left( m \right)&{=} \sqrt {{P_{\rm{i}}}} {e^{j\left( {{\theta _{m,\rm{t}}}} \right)}}
	\text{,}
	\vspace*{-2pt} %留空白，可自己调整
	\end{align}
	where ${\theta _{m,\rm{t}}} \sim {U}\left( {-\pi,\pi } \right)$.
	
	{{When $N_{\rm{i}}>1$, corresponding to multiple-tone interference, the tone interference in the DAFT domain converges in distribution to a complex centered Gaussian distribution, i.e., 
	\begin{equation}
	\vspace*{-2pt} %留空白，可自己调整
	{
		\label{eq:STI_DAFTg1}
		J_{\rm{t}}^{\rm{A}}\left( m \right) \sim \mathcal {CN}\left( {0\text{,} P_{\rm{i}}} \right)\text{.} }
	\vspace*{-2pt} %留空白，可自己调整
	\end{equation}}}
\end{proposition}
\begin{proof}
	See Appendix \ref{proof_Corollary1}.
\end{proof}
\begin{figure}
	\centering
	\includegraphics[width=2.8in]{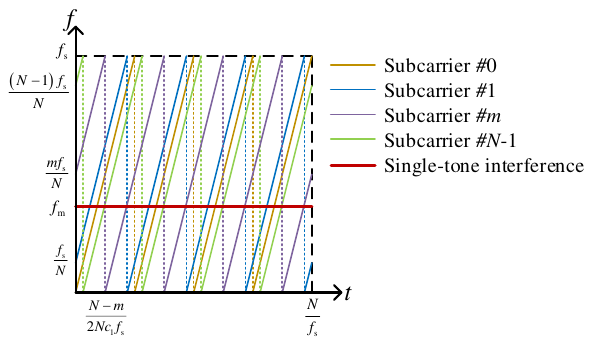}
	\vspace*{-8pt} %留空白，可自己调整
	\caption{{{Time-frequency representation of AFDM subcarriers and singe-tone interference.}}  
		\label{fg:TF_STI}} 
	\vspace*{-10pt} %留空白，可自己调整
\end{figure}
\begin{figure}
	\centering
	\includegraphics[width=2.8in]{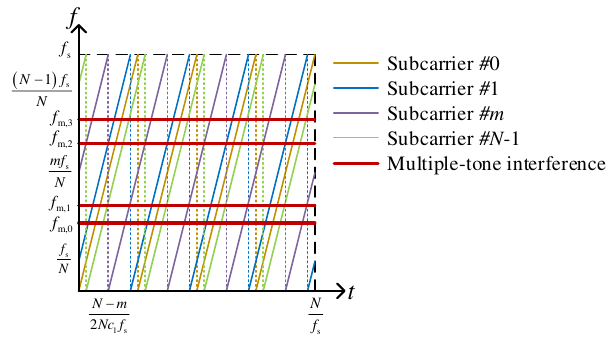}
	\vspace*{-8pt} %留空白，可自己调整
	\caption{{{Time-frequency representation of AFDM subcarriers and multiple-tone interference.}}  
		\label{fg:TF_MTI}} 
	\vspace*{-10pt} %留空白，可自己调整
\end{figure}
Proposition 1 intuitively reveals that tone interference imposes identical impact on every AFDM subcarrier. This conclusion is fully consistent with the time-frequency characteristics of both AFDM and tone interference, as illustrated in Fig. \ref{fg:TF_STI} and Fig. \ref{fg:TF_MTI}. {{When $N_{\rm{i}}=1$, corresponding to single-tone interference, each AFDM subcarrier spans the entire time–frequency plane and therefore exhibits the same degree of overlap with the single-tone component. Consequently, the resulting DAFT domain single-tone interference has identical amplitude $\sqrt {{P_{\rm{i}}}}$ across all indices. The phase term ${\theta _{m,\rm{t}}}$ in (\ref{eq:STI_DAFT}) represents the phase contributed by the interference at the overlap position of each subcarrier, and because the interference phase $\theta_{\rm{i}}$ is uniformly distributed over $\left[-\pi , \pi\right]$, the resulting ${\theta _{m,\rm{t}}}$ for each DAFT index also follows a uniform distribution over $\left[-\pi , \pi\right]$. When $N_{\rm{i}}>1$, corresponding to multiple-tone interference, each AFDM subcarrier still overlaps equally with the set of interference tones in the time–frequency plane. Since the tones have independent random phases, their DAFT domain contributions superpose to yield a random sum whose statistical behavior converges to a complex zero-mean centered Gaussian distribution. In turn, the DAFT domain multiple-tone interference exhibits Gaussian randomness across all indices, as shown in (\ref{eq:STI_DAFTg1}).}}

\subsection{Impact Analysis of Sweeping Interference}
Similar to the analysis of tone interference, the finite quadratic exponential summation significantly complicates the analysis of sweeping interference after sampling. Thus, we first analyze sweeping interference in the continuous AFT domain which is denoted by the AFT of the analog-to-digital converted sweeping interference. Then a closed-form expression of sweeping interference in the DAFT domain is derived by calculating the amplitude and phase of the sampled results of the impact in the AFT domain.

In this subsection, we consider the scenario where sweeping interference impact is maximized, i.e., the sweep interference bandwidth matches the communication signal bandwidth. In this case, $B_{\rm{i}}$ is equal to the communication bandwidth. Then the analog-to-digital converted sweeping interference is given by
\begin{align}\label{eq:SWI_ad}
\vspace*{-2pt} %留空白，可自己调整
{s_{sw}}\left( t \right)  &{{=}}\sqrt {{P_{\rm{i}}}} {e^{j{\theta _{\rm{i}}}}}\sum\limits_{{n_{\rm{s}}} = 0}^{{N_{\rm{s}}}} {\sum\limits_{n = 0}^{\frac{{{n_{\rm{s}}}N}}{{{N_{\rm{s}}}}} - 1} {{e^{j\pi \left( {{f_{\rm{m}}}{t^{'}} + {\varphi _{\rm{i}}}{{\left( {{t^{'}}} \right)}^2}} \right)}}} \delta \left( {t - \frac{n}{{{f_{\rm{s}}}}}} \right)}   \nonumber\\
&{=} \sqrt {{P_{\rm{i}}}} {e^{j{\theta _{\rm{i}}}}}\sum\limits_{n = 0}^{N - 1} {{e^{j\pi \left( {{f_{\rm{m}}}t + {\varphi _{\rm{i}}}{t^2}} \right)}}} \delta \left( {t - \frac{n}{{{f_{\rm{s}}}}}} \right)\nonumber\\
&{=} \sqrt {{P_{\rm{i}}}} {e^{j{\theta _{\rm{i}}}}}{f_{\rm{sw}}}\left( t \right)q\left( t \right)
\text{,}
\vspace*{-2pt} %留空白，可自己调整
\end{align}
where $N_{\rm{s}}\in \mathbb{N} $ denotes the number of cycles of sweeping interference within ${N}/{{{f_{\rm{s}}}}}$, $t^{'}=t - {{n_{\rm{s}}}N}/\left({{{N_{\rm{s}}}{f_{\rm{s}}}}}\right)$ and ${f_{\rm{sw}}}\left( t \right)={{e^{j\pi \left( {{f_{\rm{m}}}t + {\varphi _{\rm{i}}}{t^2}} \right)}}}g\left( {t{f_{\rm{s}}}/N} \right)$.

Then we reformulate sweeping interference in the DAFT domain, i.e., $J_{\rm{sw}}^{\rm{A}}\left( m \right)$, as
\begin{equation}
\vspace*{-2pt} %留空白，可自己调整
{
	\label{eq:SWI_DAFT2}
	J_{{\rm{sw}}}^{\rm{A}}\left( m \right) = \sqrt {\frac{{2\pi b{P_{\rm{i}}}}}{N}}  {e^{j{\theta _{\rm{i}}}}}L_{\rm{sw}}^{\rm{A}}\left( u \right){|_{u = m\Delta u}} \text{,} }
\vspace*{-2pt} %留空白，可自己调整
\end{equation}
where $L_{\rm{sw}}^{\rm{A}}\left( u \right)$ is the AFT of $\left[f_{\rm{sw}}\left(t\right)q\left(t\right)\right]$. Based on the AFT convolution theorem and Lemma 1, we obtain the closed-form expression of sweeping interference in the DAFT domain, which is provided in Proposition 2. 
\begin{proposition}
	Sweeping interference in the DAFT domain can be written as
	{{
			\begin{align}\label{eq:SWI_DAFT}
			\vspace*{-2pt} %留空白，可自己调整
			J_{{\rm{sw}}}^{\rm{A}}\left( m \right)  =
			\begin {cases}
			{\sqrt {{P_{\rm{i}}}} {e^{j{\theta _{m,\rm{w}}}}}} \quad &{{{\varphi _{\rm{i}}} \ne \frac{d}{{2\pi b}}}} \\
			{\sqrt {{NP_{\rm{i}}}} {e^{j{\theta _{m,\rm{w}}}}}\delta \left( {m - \alpha } \right)} \quad &{{{\varphi _{\rm{i}}} = \frac{d}{{2\pi b}}}} \\
			\end{cases},
			%\text{,}
			\vspace*{-2pt} %留空白，可自己调整
			\end{align}}}
	where $\alpha  =  {N{{\left\langle {\frac{{{f_{\rm{m}}}}}{{{f_{\rm{s}}}}}} \right\rangle }_1}} $ and ${\theta _{m,\rm{w}}} \sim {U}\left({-\pi,\pi } \right) $. 
\end{proposition}
\begin{proof}
	See Appendix \ref{proof_Corollary2}.
\end{proof}

\begin{figure}
	\centering
	\includegraphics[width=2.8in]{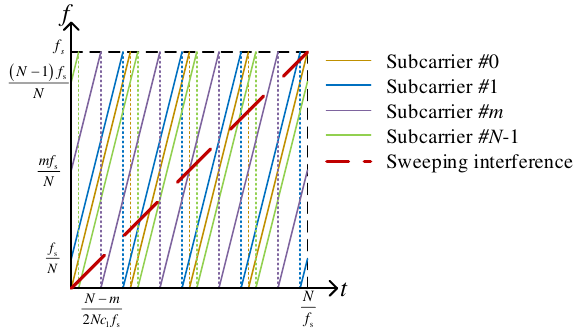}
	\vspace*{-8pt} %留空白，可自己调整
	\caption{{{Time-frequency representation of AFDM subcarriers and sweeping interference with frequency modulation slope differing from that of AFDM subcarrier.}}
		\label{fg:TF_SWI1}} 
	\vspace*{-15pt} %留空白，可自己调整
\end{figure}

\begin{figure}
	\centering
	\includegraphics[width=2.8in]{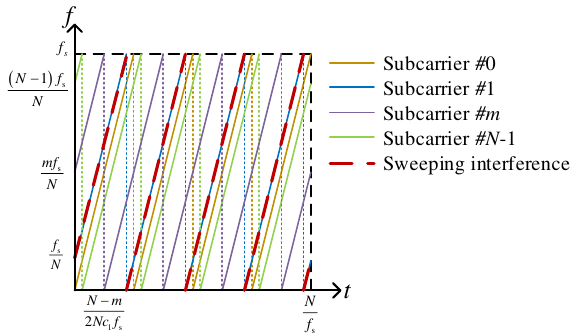}
	\vspace*{-8pt} %留空白，可自己调整
	\caption{{{Time-frequency representation of AFDM subcarriers and sweeping interference with frequency modulation slope matching that of AFDM subcarrier.}}
		\label{fg:TF_SWI2}} 
	\vspace*{-15pt} %留空白，可自己调整
\end{figure}

In Proposition 2, the condition ${{{\varphi _{\rm{i}}} \ne \frac{d}{{2\pi b}}}}$ indicates a mismatch between the frequency modulation slope of sweeping interference and that of AFDM subcarriers. When this condition is satisfied, sweeping interference imparts identical impact on each AFDM subcarrier. As shown in Fig. \ref{fg:TF_SWI1}, each AFDM subcarrier overlaps with sweeping interference to the same extent in the time-frequency plane. {{Consequently, the resulting DAFT domain sweeping interference has identical amplitude $\sqrt {{P_{\rm{i}}}}$ across all indices. And the phase term ${\theta _{m,\rm{w}}}$ in (\ref{eq:SWI_DAFT}) represents the phase contributed by the interference at the overlap position of each subcarrier.}} Whereas, ${{{\varphi _{\rm{i}}} = \frac{d}{{2\pi b}}}}$ indicates that the frequency modulation slope of interference matches that of AFDM subcarriers. In this case, the interference energy is concentrated on a specific subcarrier. As shown in Fig. \ref{fg:TF_SWI2}, {{the sweeping interference overlaps in the time-frequency plane only with the subcarrier whose initial frequency coincides with $f_{\rm{m}}$. Therefore, the resulting sweeping interference in the DAFT domain appears only at specific index, which correspond precisely to the subcarrier that coincides with the sweeping interference.}}

\subsection{Impact Analysis of Broadband Interference}
Broadband interference after down conversion and sampling at communication {{receiver}} could be modeled as
\begin{equation}
\vspace*{-2pt} %留空白，可自己调整
{
	\label{eq:BBI_n}
	{J_{\rm{bb}}}\left( n \right) = \sqrt {{P_{\rm{i}}}} z\left( n \right)\text{,} }
\vspace*{-2pt} %留空白，可自己调整
\end{equation}
where $z\left( n \right)$ denotes samples conformed to a complex centered Gaussian distribution, i.e., $z\left( n \right) \sim \mathcal {CN}\left( {0\text{,} 1} \right)$. From (\ref{eq:BBI_n}), after the DAFT, broadband interference in the DAFT domain can be written as
\begin{equation}
\vspace*{-2pt} %留空白，可自己调整
{
	\label{eq:BBI_DAFT}
	J_{\rm{bb}}^{\rm{A}}\left( m \right) = \sqrt {\frac{{{P_{\rm{i}}}}}{N}} {e^{ - j2\pi {c_2}{m^2}}}\sum\limits_{n = 0}^{N - 1} {z\left( n \right){e^{ - j2\pi \left( {\frac{{mn}}{N} + {c_1}{n^2}} \right)}}}\text{.} }
\vspace*{-2pt} %留空白，可自己调整
\end{equation}

According to (\ref{eq:DAFT}) and (\ref{eq:IDAFT}), we could get that the DAFT is a unitary transformation. Considering that unitary transformations do not change the properties of Gaussian processes \cite{DTSP}, we obtain Proposition 3.

\begin{proposition}
	Broadband interference in the DAFT domain {{conforms}} to a complex centered Gaussian distribution, which could be given by
	\begin{equation}
	\vspace*{-2pt} %留空白，可自己调整
	{
		\label{eq:BBI_effect}
		J_{\rm{bb}}^{\rm{A}}\left( m \right) \sim \mathcal {CN}\left( {0\text{,} P_{\rm{i}}} \right)\text{.} }
	\vspace*{-2pt} %留空白，可自己调整
	\end{equation}
\end{proposition}	

\begin{figure}
	\centering
	\includegraphics[width=2.8in]{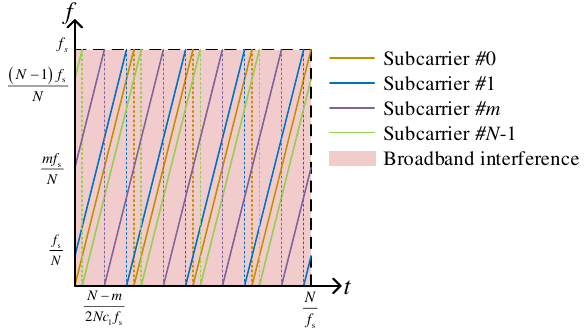}
	\vspace*{-8pt} %留空白，可自己调整
	\caption{{{Time-frequency representation of AFDM subcarriers and broadband interference.}}
		\label{fg:TF_BBI}} 
	\vspace*{-15pt} %留空白，可自己调整
\end{figure}

Based on Proposition 3, the impact of broadband interference in the DAFT domain is equivalent to {{the}} impact of additive Gaussian noise. {{As shown in Fig. \ref{fg:TF_BBI}, each AFDM subcarrier overlaps with broadband interference to the same extent in the time-frequency plane. Due to the independence and whiteness of the broadband interference in the time-frequency plane, the DAFT domain broadband interference at each index is uncorrelated and exhibits properties similar to complex Gaussian white noise, as described in (\ref{eq:BBI_effect}).}}

\subsection{Impact Analysis of Narrowband Interference}
After down conversion, sampling, and the DAFT, narrowband interference in the DAFT domain can be given by
\begin{align}
\label{eq:NBI_DAFT}
J_{\rm{nb}}^{\rm{A}}\left( m \right) = &\sqrt {\frac{{{P_{\rm{i}}}}}{N}} {e^{ - j\left( {2\pi {c_2}{m^2} + {\theta _{\rm{i}}}} \right)}}\nonumber\\
&{\cdot}\sum\limits_{n = 0}^{N - 1} {{J_{\rm{nb}}^i}\left( n \right){e^{  j2\pi {{f_{\rm{d}}}n} }}{e^{ - j2\pi \left( {\frac{{mn}}{N} + {c_1}{n^2}} \right)}}} 
\text{,}
\end{align}
where $i=1$ or $2$, ${f_{\rm{d}}} = \frac{{{f_{\rm{i}}} - {f_{\rm{c}}}}}{{{f_{\rm{s}}}}}$, $f_{\rm{s}}$ is the sampling frequency, $f_{\rm{c}}$ denotes down conversion center frequency, and ${J_{\rm{nb}}^i}\left( n \right)$ denotes the sampled baseband narrowband interference. ${J_{\rm{nb}}^1}\left( n \right)$ corresponding to signal model (\ref{eq:NBI1}) is
\begin{equation}
\vspace*{-2pt} %留空白，可自己调整
{
	\label{eq:NBI_n1}
	J_{\rm{nb}}^1\left( n \right) = \sum\limits_{k =  - \infty }^\infty  {z\left( {n - k} \right)h\left( {{B_{\rm{i}}},k} \right)} \text{,} }
\vspace*{-2pt} %留空白，可自己调整
\end{equation}
where $h\left( {{B_{\rm{i}}},k} \right)$ is the filter coefficients which are chosen to satisfy
\begin{equation}
\vspace*{-2pt} %留空白，可自己调整
{
	\label{eq:h_k}
	\sum\limits_{k =  - \infty }^{ \infty}  {{{\left| {h\left( {{B_{\rm{i}}},k} \right)} \right|}^2} = 1} \text{.} }
\vspace*{-2pt} %留空白，可自己调整
\end{equation}
And ${J_{\rm{nb}}^2}\left( n \right)$ corresponding to signal model (\ref{eq:NBI2}) could be given by
\begin{equation}
\vspace*{-2pt} %留空白，可自己调整
{
	\label{eq:NBI_n2}
	{J_{\rm{nb}}^2}\left( n \right) = \sum\limits_{p = 0}^{\infty} {g\left( {\frac{{n - p{R_{\rm{u}}}}}{{{R_{\rm{u}}}}}} \right){a_{\rm{i}}}\left( p \right)} \text{,} }
\vspace*{-2pt} %留空白，可自己调整
\end{equation}
where ${R_{\rm{u}}} = {{2{f_{\rm{s}}}}}/{{{B_{\rm{i}}}}}$. Unlike broadband interference, the sampled narrowband interference exhibits both randomness and correlation among ${J_{\rm{nb}}^i}\left( n \right)$. To evaluate impact of narrowband interference on AFDM symbols, the statistical characteristics of $J_{\rm{nb}}^{\rm{A}}\left( m \right)$ are derived.  %Thus, the orthogonal invariance property of Gaussian vectors is inapplicable for narrowband interference analysis.
Specifically, the derived statistical characteristics are provided in Proposition 4. 
\begin{proposition}
	Regardless of whether model (\ref{eq:NBI1}) or model (\ref{eq:NBI2}) is employed, narrowband interference in the DAFT domain, i.e., $J_{\rm{nb}}^{\rm{A}}\left( m \right)$, exhibits identical expectation and variance, which is
	\begin{align}
	\vspace*{-2pt} %留空白，可自己调整
	\label{eq:E&V_NBI}
	\left\{ {\begin{array}{*{20}{c}}
		\mathbb{E}\left\{ {J_{\rm{nb}}^{\rm{A}}\left( m \right)} \right\} = 0\\
		\mathbb{V}\left\{ {J_{\rm{nb}}^{\rm{A}}\left( m \right)} \right\} = {P_{\rm{i}}}
		\end{array}} \right.
	\text{.} 
	\vspace*{-3pt} %留空白，可自己调整
	\end{align}
\end{proposition}
\begin{proof}
	See Appendix \ref{proof_Corollary4}.
\end{proof}

\begin{figure}
	\centering
	\includegraphics[width=2.8in]{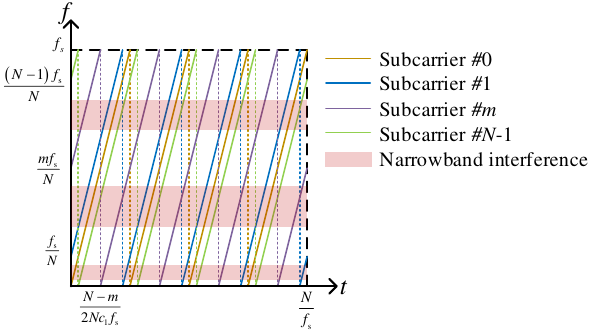}
	\vspace*{-8pt} %留空白，可自己调整
	\caption{{{Time-frequency representation of AFDM subcarriers and narrowband interference.}}
		\label{fg:TF_NBI}} 
	\vspace*{-15pt} %留空白，可自己调整
\end{figure}

From Proposition 4, it is evident that the impact of narrowband interference in the DAFT domain is independent of $m$. {{As shown in Fig. \ref{fg:TF_NBI}, each AFDM subcarrier overlaps with narrowband interference to the same extent in the time-frequency plane. Although the narrowband interference at different DAFT indices is not completely independent, its impact on each AFDM subcarrier is identical, and all subcarriers exhibit the same statistical characteristics, as shown in (\ref{eq:E&V_NBI}).}}

To summarize the impact analyses of interferences in the DAFT domain, we classify interference impacts into stationary and non-stationary categories depending on whether the impact on AFDM subcarriers is identical.

\textit{i) Stationary interference impact in the DAFT domain:} 
When interference is one of broadband, narrowband, tone, or sweeping interference (with frequency modulation slope differing from that of AFDM subcarrier), interference imparts identical impact on each AFDM subcarrier. Therefore, interference impact is referred to as stationary impact. {{Note that partial-band jamming and multi-tone interference, which often arise in practical adversarial devices, naturally fall into this category, and can be handled within the proposed anti-interference AFDM system in Section IV by treating them as stationary interference impacts.}}

\textit{ii) Non-stationary interference impact in the DAFT domain:} 
{{When interference is sweeping interference whose frequency modulation slope is aligned with that of AFDM subcarriers,}} the interference energy is concentrated on a specific AFDM subcarrier, and the interference impact is referred to as non-stationary impact.

%We define $C.2$ as  constraint that interference is one of broadband, narrowband, single-tone, or sweeping interference (with frequency modulation slope differing from that of AFDM subcarrier). When $C.2$ is satisfied, 
%
%Let $C.3$ denote constraint that interference is sweeping interference that is consistent with AFDM subcarrier frequency modulation slope. When $C.3$ is satisfied, 
\section{Packet Throughput-Guided Parameter Optimization Algorithm of Anti-interference AFDM System}

To balance resource efficiency and anti-interference capability, we derive the analytical relationship between packet throughput and the system {{parameters}} and design a parameter optimization algorithm that maximizes packet throughput. 
We analytically derive the impact of system parameters on packet throughput under different types of interference and doubly selective channel conditions. Accordingly, we propose a parameter optimization algorithm to optimize the spreading sequence length under fixed ECC parameters, where the nonlinear optimization objective is solved by Newton's method.

The classification of interference impact allows us to formulate a unified anti-interference AFDM framework as illustrated in Fig. \ref{fg:FRAMEWORK}. This framework integrates spread spectrum in the DAFT domain and ECC to suppress both stationary and non-stationary interference impact. {{After that, we focus on analyzing the relationship between the framework parameters and packet throughput and perform parameter optimization, where the proposed algorithm jointly handles both stationary and non-stationary interference impact.}}

{{For clarity, the main variables in this section are listed in Table \ref{tab:variable_tab}.
		\begin{table}
			\centering
			%%\caption{****}
			\caption{{Main Variables in Section IV}}\label{tab:variable_tab}
			\renewcommand\arraystretch{1.2}
			\begin{threeparttable}
				\begin{tabular}{|c|l|}
					\hline
					\textbf{{{Variable}}}                       &\textbf{{{Meaning}}}                       \\ \hline
					{$N_{\rm{p}}$} &  {{The number of bits per packet}}    \\ \hline
					{$N_{\rm{i}}$} &  {{Number of input bits for each ECC codeword}}     \\ \hline
					{$N_{\rm{o}}$} &  {{Number of output bits for each ECC codeword}}     \\ \hline			
					{$N_{\rm{e}}$} &  {{The maximum number of error correction bits per group}}     \\ \hline
					{$N_{\rm{m}}$} &  {{Modulation order of a constellation diagram}}   \\ \hline
					{$N_{\rm{cp}}$} &  {{The length of the CPP}}     \\ \hline			
					{$P_{\rm{s}}$} &  {{The transmitted power}}  \\ \hline			
					{$B_{\rm{c}}$} &  {{The bandwidth occupied by the communication system}}     \\ \hline			
					{${P_{\rm{dc}}}$} &  {{Codeword decoding success probability}}    \\ \hline			
					{$\eta$}  &  {{Packet throughput}}    \\ \hline				
				\end{tabular}
			\end{threeparttable}
		\end{table}
}}
\begin{figure}
	\centering
	\includegraphics[width=3.55in]{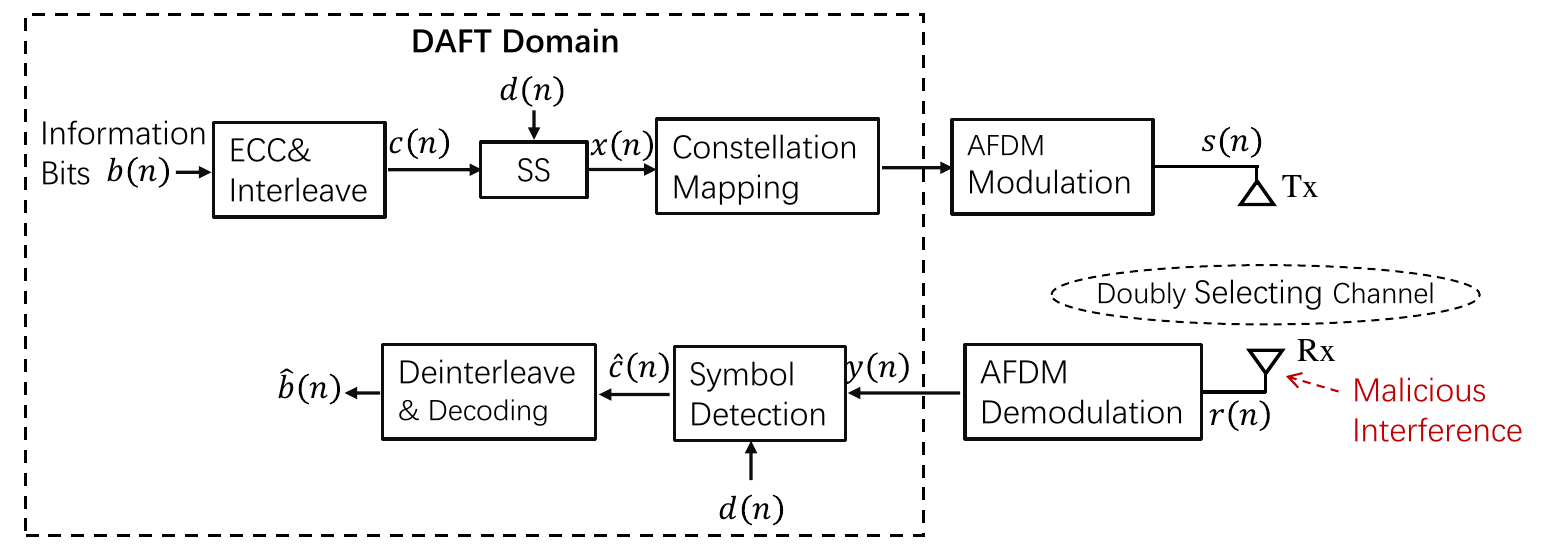}
	\vspace*{-23pt} %留空白，可自己调整
	\caption{The framework of anti-interference AFDM system.  
		\label{fg:FRAMEWORK}} 
	\vspace*{-20pt} %留空白，可自己调整
\end{figure}

\subsection{{{Design Principles of Spreading Sequences for Anti-Interference AFDM System}}}

{{Based on the analysis of interferences in the DAFT domain presented in Section III, we discuss the design principles of spreading sequences from whether the interference samples across DAFT indices are statistically independent.}}

{For broadband interference and matched-slope sweeping interference, the DAFT domain interference samples are shown to be independent across DAFT indices, as shown in Proposition 2 and Proposition 3. In this case, the role of spreading is primarily to distribute each information symbol over multiple AFDM subcarriers. Since no coherent accumulation of interference in the DAFT domain occurs during despreading, the correlation properties of the spreading sequences do not impose any fundamental constraint in this case. As a result, any spreading sequence, including repeated or arbitrary sequences, can be employed to provide interference robustness.}

In contrast, for narrowband interference, tone interference, and non-matched sweeping interference, the DAFT domain interference samples are not independent and exhibit correlation due to phase coupling. For example, in Proposition 2, although $\theta_{m,\rm{w}} \sim U\left(-\pi.\pi\right)$, $\theta_{m,\rm{w}}$ is not fully independent across $m$, leading to correlation in $J_{{\rm{sw}}}^{\rm{A}}\left( m \right)$ across DAFT indices $m$. In such case, inappropriate spreading sequences may lead to coherent accumulation of interference during despreading. To suppress this effect and to maximize the post-despreading signal to interference ratio (SIR) for a given spreading sequence length, it is necessary to ensure coherent accumulation of the desired AFDM signal while simultaneously destroying the correlation structure of the DAFT domain interference, thereby preventing coherent accumulation of the interference. This requires the spreading sequence $d\left[n\right]$ to have good autocorrelation properties, i.e., the autocorrelation function
\begin{equation}
\vspace*{-2pt} %留空白，可自己调整
{
	\label{eq:autocorrelation}
	{{R_d}\left( k \right) = \sum\limits_{n = 0}^{{N_{\rm{d}}-1}} {d\left[ n \right]} {d^*}\left[ {{{\left\langle {n - k} \right\rangle }_{{N_{\rm{d}}}}}} \right]}\text{,} }
\vspace*{-2pt} %留空白，可自己调整
\end{equation}
where $N_{\rm{d}}$ denotes the spreading sequence length, should exhibit low sidelobes, so that the interference cannot accumulate coherently after despreading irrespective of its correlation. In addition, the anti-interference performance depends only on the autocorrelation of $d\left[n\right]$, while the cross-correlation properties are irrelevant.  

{To unify the treatment of DAFT domain interference with both correlated and uncorrelated structures, and to maintain the low-complexity design of the proposed reception in Section V,the proposed anti-interference AFDM system adopts spreading sequences with good autocorrelation properties. These sequences can be chosen from m-sequences generated by Linear Feedback Shift Registers (LFSR) or from single Gold sequences. The design principle ensures that the desired AFDM symbols coherently combine after despreading, while the correlated DAFT domain interference becomes non-coherent, thereby improving SIR and enhancing the overall anti-interference capability. Under this unified treatment, we develop the parameter optimization algorithm based on the classification of interference impacts into stationary and non-stationary categories.}

\subsection{Relationship between Packet Throughput and System Parameters}
To provide a realistic basis for high-mobility scenarios, the relationship between packet throughput and the system parameters in Fig. \ref{fg:FRAMEWORK} needs to be analytically derived under doubly selective channel constraints.
The probability of successful packet transmission can be first given by ${P_{\rm{pc}}} = {P_{\rm{dc}}}^{G}$, where $G = {{{N_{\rm{p}}}} \mathord{\left/
		{\vphantom {{{N_{\rm{p}}}} {{N_{\rm{i}}}}}} \right.
		\kern-\nulldelimiterspace} {{N_{\rm{i}}}}}$ and ${P_{\rm{dc}}}$ is given by
\begin{equation}
\vspace*{-2pt} %留空白，可自己调整
{
	\label{eq:Pdc}
	{P_{\rm{dc}}} = \sum\limits_{k = 0}^{{N_{\rm{e}}}} { \binom{N_{\rm{o}}}{k}} {P_{\rm{e}}}^k{{\left( {1 - {P_{\rm{e}}}} \right)}^{{N_{\rm{o}}} - k}}  \text{,} }
\vspace*{-2pt} %留空白，可自己调整
\end{equation}
where $P_{\rm{e}}$ is the BER before decoding, which is equal to {{the}} BER after despreading. For simplicity, binary phase shift keying (BPSK) and quadrature phase shift keying (QPSK) are taken into account in this paper. In order to achieve the randomization of errors caused by interference, the interleaving depth is designed as a multiple of $N{{\log }_2}\left( {{N_{\rm{m}}}} \right)/{N_{\rm{d}}}$ in the proposed system. {{After deinterleaving, the BERs for stationary and non-stationary interference impacts, $P_{{\rm{e,s}}}$ and $P_{{\rm{e,n}}}$, are given in Proposition 5.}}
\begin{proposition}
	In the proposed anti-interference AFDM system, the BER before decoding under stationary interference impact, i.e., $P_{{\rm{e,s}}}$, can be expressed as
	\begin{equation}
	\vspace*{-2pt} %留空白，可自己调整
	{
		\label{eq:Pes}
		{P_{\rm{e,s}}} = \frac{1}{2} - \frac{1}{2}\Theta\left( {{N_{\rm{d}}}} \right)\text{,} }
	\vspace*{-2pt} %留空白，可自己调整
	\end{equation}
	where 
	\vspace*{-6pt} %留空白，可自己调整
	\begin{flalign}\label{eq:F}
	&\		\Theta \left( {{N_{\rm{d}}}} \right) = \sum\limits_{i = 0}^{L - 1} { \binom{2i}{i}} {\left( {\frac{{{\gamma _{\rm{in}}}}}{{4{N_{\rm{d}}}}}} \right)^i}{\left( {\frac{{{N_{\rm{d}}}}}{{{N_{\rm{d}}} + {\gamma _{\rm{in}}}}}} \right)^{i + 0.5}}  . &
	\end{flalign}
	Whereas, the BER before decoding under non-stationary interference impact, i.e., $P_{{\rm{e,n}}}$, can be expressed as
	\begin{equation}
	\vspace*{-2pt} %留空白，可自己调整
	{
		\label{eq:Pens}
		{P_{\rm{e,n}}} = \frac{1}{2} - \frac{1}{2}\Psi\left( {{N_{\rm{d}}}} \right)\text{,} }
	\vspace*{-2pt} %留空白，可自己调整
	\end{equation}
	where 
	\vspace*{-6pt} %留空白，可自己调整
	\begin{flalign}\label{eq:G1}
	&\		\Psi \left( {{N_{\rm{d}}}} \right) = R{\Psi _2}\left( {{N_{\rm{d}}}} \right) + \left( {1 - R} \right){\Psi _1}\left( {{N_{\rm{d}}}} \right)  , &
	\end{flalign}
	%	\vspace*{-10pt} %留空白，可自己调整
	\vspace*{-15pt} %留空白，可自己调整
	\begin{flalign}\label{eq:G}
	&\		{\Psi _1}\left( {{N_{\rm{d}}}} \right) = \sum\limits_{i = 0}^{L - 1}  { \binom{2i}{i}} {\left( {\frac{{{\gamma _{\rm{n}}}}}{{4{N_{\rm{d}}}}}} \right)^i}{\left( {\frac{{{N_{\rm{d}}}}}{{{N_{\rm{d}}} + {\gamma _{\rm{n}}}}}} \right)^{i + 0.5}}  , &
	\end{flalign}
	\vspace*{-10pt} %留空白，可自己调整
	\begin{flalign}
	\label{eq:G2} & \		{\Psi _2}\left( {{N_{\rm{d}}}} \right) = \sum\limits_{i = 0}^{L - 1}  { \binom{2i}{i}}{\left( {\frac{{{\gamma _{\rm{m}}}}}{{4{N_{\rm{d}}}^3}}} \right)^i}{\left( {\frac{{{N_{\rm{d}}}^3}}{{{N_{\rm{d}}}^3 + {\gamma _{\rm{m}}}}}} \right)^{i + 0.5}}  , &
	\end{flalign}
	${\gamma _{\rm{in}}} = {{L{{\log }_2}\left( {{N_{\rm{m}}}} \right)\left( {{P_{\rm{n}}} + {P_{\rm{i}}}} \right)} \mathord{\left/{\vphantom {{L{{\log }_2}\left( {{N_{\rm{m}}}} \right)\left( {{P_{\rm{n}}} + {P_{\rm{i}}}} \right)} {{P_{\rm{s}}}}}} \right.\kern-\nulldelimiterspace} {{P_{\rm{s}}}}}$, ${\gamma _{\rm{n}}} = {{L{{\log }_2}\left( {{N_{\rm{m}}}} \right){P_{\rm{n}}}} \mathord{\left/
			{\vphantom {{L{{\log }_2}\left( {{N_{\rm{m}}}} \right){P_{\rm{n}}}} {{P_{\rm{s}}}}}} \right.
			\kern-\nulldelimiterspace} {{P_{\rm{s}}}}}$, ${\gamma _{\rm{m}}} = {\gamma _{\rm{i}}} + {\gamma _{\rm{n}}}{N_{\rm{d}}}^2$,
	${\gamma _{\rm{i}}} = {{{N}L{{\log }_2}\left( {{N_{\rm{m}}}} \right){P_{\rm{i}}}} \mathord{\left/
			{\vphantom {{{N}L{{\log }_2}\left( {{N_{\rm{m}}}} \right){P_{\rm{i}}}} {{P_{\rm{s}}}}}} \right.
			\kern-\nulldelimiterspace} {{P_{\rm{s}}}}}$, and $R={{{{N_{\rm{d}}}} \mathord{\left/
				{\vphantom {{{N_{\rm{d}}}} {\left[ {N{{\log }_2}\left( {{N_{\rm{m}}}} \right)} \right]}}} \right.
				\kern-\nulldelimiterspace} {\left[ {N{{\log }_2}\left( {{N_{\rm{m}}}} \right)} \right]}}}$.
\end{proposition}
\begin{proof}
	See Appendix \ref{proof_Corollary5}.
\end{proof}

{{For notational convenience and consistency, $F\left( {{N_{\rm{d}}}}\right) $ is used to represent both $\Theta \left( {{N_{\rm{d}}}} \right)$ and $\Psi \left( {{N_{\rm{d}}}} \right)$, which is defined as}}
\begin{align}\label{eq:PeResult}
\vspace*{-2pt} %留空白，可自己调整
F\left( {{N_{\rm{d}}}}\right) = 
\begin {cases}
\Theta \left( {{N_{\rm{d}}}} \right) \quad &{C.2} \\
\Psi \left( {{N_{\rm{d}}}} \right)  \quad &{C.3} \\
\end{cases},
%\text{,}
\vspace*{-2pt} %留空白，可自己调整
\end{align}
where $C.2$ is the constraint that interference is one of broadband, narrowband, single-tone, or sweeping interference (with frequency modulation slope differing from that of AFDM subcarrier), 
$C.3$ denotes the constraint that interference is sweeping interference with a frequency modulation slope consistent with that of AFDM subcarrier.	

From (\ref{eq:PT_de}), packet throughput of the anti-interference AFDM system can thereby be given by Corollary 1.
\begin{corollary}
	The relationship between packet throughput and the system {{parameters}} can be indicated by
	\begin{align}\label{eq:PT}
	\vspace*{-2pt} %留空白，可自己调整
	&\ \eta  = \frac{{{{\left[ {\sum\limits_{k = 0}^{{N_{\rm{e}}}} { \binom{N_{\rm{o}}}{k}} {{\left( {1 - F\left( {{N_{\rm{d}}}} \right)} \right)}^k}{{\left( {1 + F\left( {{N_{\rm{d}}}} \right)} \right)}^{{N_{\rm{o}}} - k}}} \right]}^G}}}{{{2^{GN_{\rm{o}}}}K{N_{\rm{d}}}}}, &
	\vspace*{-2pt} %留空白，可自己调整
	\end{align}
	where ${{K = {N_{\rm{p}}}{N_{\rm{o}}}\left( {N + {N_{\rm{cp}}}} \right)} \mathord{\left/
			{\vphantom {{K = {N_{\rm{p}}}{N_{\rm{o}}}\left( {N + {N_{\rm{cp}}}} \right)} {\left( {N{N_{\rm{i}}}{B_{\rm{c}}}{{\log }_2}\left( {{N_{\rm{m}}}} \right)} \right)}}} \right.
			\kern-\nulldelimiterspace} {\left[ {N{N_{\rm{i}}}{B_{\rm{c}}}{{\log }_2}\left( {{N_{\rm{m}}}} \right)} \right]}}$.
\end{corollary}

\subsection{Parameter Optimization Algorithm of Anti-interference System}

Considering the limited flexibility in ECC parameters and the significant impact of $N_{\rm{d}}$ on $\eta$, parameter optimization is primarily oriented toward optimizing $N_{\rm{d}}$. The function $\eta$ is clearly convex with respect to $N_{\rm{d}}$. Then, the optimal $N_{\rm{d}}$ that maximizes $\eta$ satisfies 
\begin{equation}
\vspace*{-2pt} %留空白，可自己调整
{
	\label{eq:d0}
	U\left(N_{\rm{d}}\right)=\frac{{d\eta }}{{d{N_{\rm{d}}}}} = 0  \text{.} }
\vspace*{-2pt} %留空白，可自己调整
\end{equation}
Combining (\ref{eq:PT}), (\ref{eq:Pdc}) and (\ref{eq:PeResult}),  $U\left(N_{\rm{d}}\right)$ is reformulated as
\begin{align}\label{eq:U}
\vspace*{-2pt} %留空白，可自己调整
U\left( {{N_{\rm{d}}}} \right)=&\frac{{{F_1}\left( {{N_{\rm{d}}}} \right){N_{\rm{p}}}{N_{\rm{d}}}{N_{\rm{o}}}{ \binom{N_{\rm{o}}-1}{N_{\rm{e}}}}}}{{{N_{\rm{i}}}}}\nonumber\\
&{\cdot}\sum\limits_{k = 0}^{{N_{\rm{e}}}} {{ \binom{N_{\rm{o}}}{k}}} {\left[ {1 - F\left( {{N_{\rm{d}}}} \right)} \right]^{k - {N_{\rm{e}}}}}{\left[ {1 + F\left( {{N_{\rm{d}}}} \right)} \right]^{1 - k + {N_{\rm{e}}}}}
\text{,}
\vspace*{-2pt} %留空白，可自己调整
\end{align}
where ${F_1}\left( {{N_{\rm{d}}}} \right)$ denotes the first derivative of ${F}\left( {{N_{\rm{d}}}} \right) $ with respect to $N_{\rm{d}}$. Objective (\ref{eq:U}) is evidently a nonlinear equation, which poses significant {{challenges}} for obtaining an analytical solution. Thus, we propose a parameter optimization algorithm based on {{the}} bisection method and Newton’s method. \textit{First}, to ensure convergence speed, the first-order convergent bisection method is employed to determine a suitable initial point for the second-order convergent Newton's method. 
{{Because $\eta$ increases with $N_{\rm{d}}$ when the despreading signal to interference plus noise ratio (SINR) is low and decreases when $N_{\rm{d}}$ becomes excessively large, the optimal $N_{\rm{d}}$ corresponds to the SINR interval $\left[-9,3\right]$ dB \cite{FWC}. Consistent with prior works such as \cite{chow2002practical}, we assume that the noise power $P_{\rm{n}}$ and the interference power $P_{\rm{i}}$ are available for the purpose of parameter optimization. Thus,}}
the feasible range of $N_{\rm{d}}$ is defined as
\begin{align}\label{eq:bisection}
\vspace*{-2pt} %留空白，可自己调整
{\mathbb{B}} = 
\begin {cases}
\left[ {\left\lfloor {\frac{{{\gamma _{\rm{in}}}}}{8}} \right\rfloor ,\left\lceil {2{\gamma _{\rm{in}}}} \right\rceil } \right] \quad &{C.2} \\
\left[ {\left\lfloor {\frac{{{\gamma _{\rm{n}}}}}{8}} \right\rfloor ,\left\lceil {2{\gamma _{\rm{n}}}} \right\rceil } \right]  \quad &{C.3} \\
\end{cases}.
%\text{,}
\vspace*{-2pt} %留空白，可自己调整
\end{align}
\textit{Second}, the change in convergence direction is adopted as the stopping criterion, instead of relying on a predefined threshold. 
To utilize Newton's method, the first derivative of ${U}\left( {{N_{\rm{d}}}} \right) $ with respect to $N_{\rm{d}}$ as (\ref{eq:dU}) is given on the next page,
\begin{figure*}[ht]
	\vspace*{-10pt} %留空白，可自己调整
	\small
	%\begin{flalign}\label{eq:LCT_SWI}
	\begin{flalign}\label{eq:dU}
	\vspace*{-2pt} %留空白，可自己调整
	&\ \frac{{dU\left( {{N_{\rm{d}}}} \right)}}{{d{N_{\rm{d}}}}} = \frac{{{N_{\rm{p}}}{N_{\rm{o}}}{ \binom{N_{\rm{o}}-1}{N_{\rm{e}}}}}}{{{N_{\rm{i}}}}}\left[ {{N_{\rm{d}}}{F_1}\left( {{N_{\rm{d}}}} \right) + {F_2}\left( {{N_{\rm{d}}}} \right)} \right] &\nonumber\\
	&\ \quad \quad \quad \quad \quad {+} \sum\limits_{k = 0}^{{N_{\rm{e}}}} {{ \binom{N_{\rm{o}}}{k}}} {\left[ {1 - F\left( {{N_{\rm{d}}}} \right)} \right]^{k - {N_{\rm{e}}} - 1}}{\left[ {1 + F\left( {{N_{\rm{d}}}} \right)} \right]^{1 - k + {N_{\rm{e}}} - 1}}\left[ {\left( {1 - 2k} \right)F\left( {{N_{\rm{d}}}} \right){F_1}\left( {{N_{\rm{d}}}} \right) - \left( {1 - 2k + 2{N_{\rm{e}}}} \right){F_1}\left( {{N_{\rm{d}}}} \right)} \right]. &
	%\text{,}
	\vspace*{-2pt} %留空白，可自己调整
	\end{flalign}
	%\end{flalign}\normalsize
	\hrule
	\vspace*{-10pt} %留空白，可自己调整
\end{figure*}
where ${F_2}\left( {{N_{\rm{d}}}} \right)$ denotes the second derivative of ${F}\left( {{N_{\rm{d}}}} \right) $ with respect to $N_{\rm{d}}$. ${F_1}\left( {{N_{\rm{d}}}} \right)$ and ${F_2}\left( {{N_{\rm{d}}}} \right)$ is given by 
\begin{align}\label{eq:Fi}
\vspace*{-2pt} %留空白，可自己调整
{F_l}\left( {{N_{\rm{d}}}} \right) = 
\begin {cases}
\sum\limits_{i = 0}^{L - 1} { \binom{2i}{i}} {\left( {\frac{{{\gamma _{\rm{in}}}}}{4}} \right)^i}{\beta _l}\left( {{N_{\rm{d}}},i} \right) \quad &{C.2} \\
\frac{1}{{N{{\log }_2}\left( {{N_{\rm{m}}}} \right)}}\sum\limits_{i = 0}^{L - 1} { \binom{2i}{i}}{z _l}\left( {{N_{\rm{d}}},i} \right)    &{C.3} \\
\end{cases},
%\text{,}
\vspace*{-2pt} %留空白，可自己调整
\end{align}
where ${z_l}\left( {{N_{\rm{d}}},i} \right) = {\left( {\frac{{{\gamma _{\rm{n}}}}}{4}} \right)^i}{\phi _l}\left( {{N_{\rm{d}}},i} \right) + {\mu _l}\left( {{N_{\rm{d}}},i} \right)$, $l=1,2$,
\vspace*{-2pt} %留空白，可自己调整
\begin{flalign}\label{eq:beta1}
&\	{\beta _1}\left( {{N_{\rm{d}}},i} \right) = \frac{{{\gamma _{\rm{in}}} - 2i{N_{\rm{d}}}}}{{2{N_{\rm{d}}}^{0.5}{{\left( {{N_{\rm{d}}} + {\gamma _{\rm{in}}}} \right)}^{i + 1.5}}}}  , &
\end{flalign}	
\vspace*{-10pt} %留空白，可自己调整
\begin{flalign}\label{eq:beta2}
&\	{\beta _2}\left( {{N_{\rm{d}}},i} \right) = \frac{{4\left( {i + {i^2}} \right)N_{\rm{d}}^2 - 4\left( {i + 1} \right){\gamma _{\rm{in}}}{N_{\rm{d}}} - {\gamma _{\rm{in}}}^2}}{{4N_{\rm{d}}^{1.5}{{\left( {{N_{\rm{d}}} + {\gamma _{\rm{in}}}} \right)}^{i + 2.5}}}}, &
\end{flalign}
\vspace*{-15pt} %留空白，可自己调整
\begin{flalign}\label{eq:phi1}
&\	{\phi _1}\left( {{N_{\rm{d}}},i} \right) = {\left( {{N_{\rm{d}}} + {\gamma _{\rm{n}}}} \right)^{ - \left( {i + 1.5} \right)}}\psi \left( {{N_{\rm{d}}},i} \right)  , &
\end{flalign}
\vspace*{-25pt} %留空白，可自己调整
\begin{flalign}\label{eq:phi11}
&\ \psi \left( {{N_{\rm{d}}},i} \right) = \left( {i - 1} \right){N_{\rm{d}}}^{1.5} - \left[ {iN{{\log }_2}\left( {{N_{\rm{m}}}} \right) + \frac{3}{2}{\gamma _{\rm{n}}}} \right]{N_{\rm{d}}}^{0.5} & \nonumber\\
&\ \quad \quad \quad \quad \quad {+}\frac{{{\gamma _{\rm{n}}}N{{\log }_2}\left( {{N_{\rm{m}}}} \right)}}{2}{N_{\rm{d}}}^{ - 0.5},&
\end{flalign}
and the definitions of ${\phi _2}\left( {{N_{\rm{d}}},i} \right)$, ${\mu _1}\left( {{N_{\rm{d}}},i} \right)$ and ${\mu _2}\left( {{N_{\rm{d}}},i} \right)$ are provided in (\ref{eq:phi2}), (\ref{eq:mu1}) and (\ref{eq:mu2}) on the next page, respectively. Then the update of $N_{\rm{d}}$ can be expressed as
\begin{equation}
\vspace*{-2pt} %留空白，可自己调整
{
	\label{eq:Ndn}
	{N_{{\rm{d}},n + 1}} = {N_{{\rm{d}},n}} + D\left( {{N_{{\rm{d}},n}}} \right)\text{,} }
\vspace*{-2pt} %留空白，可自己调整
\end{equation}
where 
\begin{equation}
\vspace*{-2pt} %留空白，可自己调整
{
	\label{eq:Delta_Nd}
	D\left( {{N_{{\rm{d}},n}}} \right)=   - \left\lceil {U\left( {{N_{{\rm{d}},n}}} \right)\frac{{d{N_{\rm{d}}}}}{{dU\left( {{N_{\rm{d}}}} \right)}}{|_{{N_{\rm{d}}} = {N_{{\rm{d}},n}}}}} \right\rceil \text{.} }
\vspace*{-2pt} %留空白，可自己调整
\end{equation}
\begin{figure*}[htbp]
	\vspace*{-10pt} %留空白，可自己调整
	\small
	%\begin{flalign}\label{eq:LCT_SWI}
	\begin{flalign}\label{eq:phi2}
	\vspace*{-2pt} %留空白，可自己调整
	&\ {\phi _2}\left( {{N_{\rm{d}}},i} \right) = \frac{{{{\left( {{N_{\rm{d}}} + {\gamma _{\rm{n}}}} \right)}^{ - \left( {i + 1.5} \right)}}}}{4} \Bigg\{ {{ - \left( {4{i^2} + 8i + 6} \right){N_{\rm{d}}}^{1.5} + \left[ {4\left( {{i^2} + i} \right)N{{\log }_2}\left( {{N_{\rm{m}}}} \right) + 12i + 6\left( {1 - {\gamma _{\rm{n}}}} \right)} \right]{N_{\rm{d}}}^{0.5}}} &\nonumber\\
	&\ \quad \quad \quad \quad \quad \quad \quad \quad \quad \quad \quad \quad \quad \quad {{ - \left[ {4\left( {i + 1} \right)N{{\log }_2}\left( {{N_{\rm{m}}}} \right) - 12i - 6 + 3{\gamma _{\rm{n}}}} \right]{\gamma _{\rm{n}}}{N_{\rm{d}}}^{ - 0.5} - {\gamma _{\rm{n}}}^2N{{\log }_2}\left( {{N_{\rm{m}}}} \right){N_{\rm{d}}}^{ - 1.5}}} \Bigg\}.&
	%\text{,}
	\vspace*{-2pt} %留空白，可自己调整
	\end{flalign}
	%\end{flalign}\normalsize
	\vspace*{-10pt} %留空白，可自己调整
\end{figure*}
\begin{figure*}[htbp]
	\vspace*{-10pt} %留空白，可自己调整
	\small
	%\begin{flalign}\label{eq:LCT_SWI}
	\begin{flalign}\label{eq:mu1}
	\vspace*{-2pt} %留空白，可自己调整
	&\ {\mu _1}\left( {{N_{\rm{d}}},i} \right) = \frac{{{{\left( {{N_{\rm{d}}}^3 + {\gamma _{\rm{i}}} + {\gamma _{\rm{n}}}{N_{\rm{d}}}^2} \right)}^{ - \left( {i + 1.5} \right)}}}}{8}{\left( {\frac{{{\gamma _{\rm{i}}} + {\gamma _{\rm{n}}}{N_{\rm{d}}}^2}}{4}} \right)^{i - 1}} &\nonumber\\
	&\ \quad \quad \quad \quad \quad \quad  {\cdot}\left[ {\left( {2 - 2i} \right){\gamma _{\rm{n}}}{N_{\rm{d}}}^{6.5} + 3{\gamma _{\rm{n}}}^2{N_{\rm{d}}}^{5.5} + \left( {2 - 6i} \right){\gamma _{\rm{i}}}{N_{\rm{d}}}^{4.5} + 8{\gamma _{\rm{i}}}{\gamma _{\rm{n}}}{N_{\rm{d}}}^{3.5} + 5{\gamma _{\rm{i}}}^2{N_{\rm{d}}}^{1.5}} \right].&
	%\text{,}
	\vspace*{-2pt} %留空白，可自己调整
	\end{flalign}
	%\end{flalign}\normalsize
	\vspace*{-10pt} %留空白，可自己调整
\end{figure*}
\begin{figure*}[htbp]
	\vspace*{-10pt} %留空白，可自己调整
	\small
	%\begin{flalign}\label{eq:LCT_SWI}
	\begin{flalign}\label{eq:mu2}
	\vspace*{-2pt} %留空白，可自己调整
	&\ {\mu _2}\left( {{N_{\rm{d}}},i} \right) =  {\left( {\frac{{{\gamma _{\rm{i}}} + {\gamma _{\rm{n}}}{N_{\rm{d}}}^2}}{4}} \right)^i}{\left( {{N_{\rm{d}}}^3 + {\gamma _{\rm{i}}} + {\gamma _{\rm{n}}}{N_{\rm{d}}}^2} \right)^{ - \left( {i + 0.5} \right)}}{N_{\rm{d}}}^{2.5} & \nonumber\\
	&\ \quad \quad \quad \quad \quad \quad{\cdot}\Bigg[{\frac{{i{\gamma _{\rm{n}}}({\gamma _{\rm{i}}} - {\gamma _{\rm{n}}}N_{\rm{d}}^2)}}{{{{({\gamma _{\rm{i}}} + {\gamma _{\rm{n}}}N_{\rm{d}}^2)}^2}}} - \frac{5}{{2N_{\rm{d}}^2}} + \frac{{(2i + 1)\left( {3N_{\rm{d}}^4 + 4{\gamma _{\rm{n}}}N_{\rm{d}}^3 + 2{\gamma _{\rm{n}}}^2{N_{\rm{d}}}^2 - 6{\gamma _{\rm{i}}}{N_{\rm{d}}} - 2{\gamma _{\rm{i}}}{\gamma _{\rm{n}}}} \right)}}{{2{{(N_{\rm{d}}^3 + {\gamma _{\rm{i}}} + {\gamma _{\rm{n}}}N_{\rm{d}}^2)}^2}}}} & \nonumber\\
	&\ \quad \quad \quad \quad \quad \quad \quad  +{  {{\left( {\frac{{i{\gamma _{\rm{n}}}{N_{\rm{d}}}}}{{{\gamma _{\rm{i}}} + {\gamma _{\rm{n}}}N_{\rm{d}}^2}} + \frac{5}{{2{N_{\rm{d}}}}} - \frac{{(2i + 1)(3N_{\rm{d}}^2 + 2{\gamma _{\rm{n}}}{N_{\rm{d}}})}}{{2(N_{\rm{d}}^3 + {\gamma _{\rm{i}}} + {\gamma _{\rm{n}}}N_{\rm{d}}^2)}}} \right)}^2}} \Bigg]. &
	%\text{,}
	\vspace*{-2pt} %留空白，可自己调整
	\end{flalign}
	%\end{flalign}\normalsize
	\hrule
	\vspace*{-10pt} %留空白，可自己调整
\end{figure*}

\begin{algorithm}[tb]
	\caption{ Packet throughput-guided parameter optimization algorithm }\label{algo:Design_Nd}

	\textbf{Iuput}: ${{\gamma _{\rm{n}}}}$, ${{\gamma _{\rm{i}}}}$, ${{\gamma _{\rm{in}}}}$, $N$, $N_{\rm{cp}}$, $N_{m}$, $N_{p}$, $N_{i}$, $N_{o}$, $N_{e}$.
	
	\textbf{Initialize}: Let $n=0$. Set $N_{d,0}$ based on the bisection method and (\ref{eq:bisection}).

	\quad \textbf{repeat}
	
	\qquad 1: \ Compute ${F_1}\left( {{N_{{\rm{d}},n}}} \right)$ and ${F_2}\left( {{N_{{\rm{d}},n}}} \right)$ based on (\ref{eq:Fi}) to (\ref{eq:mu2}).
	
	\qquad 2: \ Calculate ${U\left( {{N_{{\rm{d}},n}}} \right)}$ and $\frac{{d{N_{\rm{d}}}}}{{dU\left( {{N_{\rm{d}}}} \right)}}{|_{{N_{\rm{d}}} = {N_{{\rm{d}},n}}}}$ based on (\ref{eq:U}) and (\ref{eq:dU}), respectively.
	
	\qquad 3: \ Derive $D\left( {{N_{{\rm{d}},n}}} \right)$ by (\ref{eq:Delta_Nd}).
	
	\qquad 4: \ Determine whether the iteration direction has changed based on $\Delta D = D\left( {{N_{{\rm{d}},n}}} \right)D\left( {{N_{{\rm{d}},n - 1}}} \right)$.
	
	\qquad 5: \ Update ${N_{{\rm{d}},n + 1}}$ using (\ref{eq:Ndn}).
	
	\qquad 6: \ Set ${n=n+1}$.
	
	\quad \textbf{until}  $\Delta D \le 0$.
	
	\textbf{Output}: Optimized $N_{\rm{d}}={N_{{\rm{d}},n-1}}$.	
	
\end{algorithm} 
Consequently, the procedure to optimize $N_{\rm{d}}$ is summarized as Algorithm \ref{algo:Design_Nd}.

\section{Low-Complexity Reception Design}
In this section, a linear-complexity CDD is developed for the proposed anti-interference AFDM system which could achieve full diversity gain.
By jointly utilizing the impulse-like autocorrelation function of spreading sequences and the cyclic-shift property of AFDM input-output relation, each path can be equalized individually without inter-path interference. This enables us to develop a linear-complexity CDD by utilizing correlation-based processing to avoid the matrix inversion operation.

\subsection{Correlation-Based DAFT Domain Detector}

Let $k_i=\alpha_{{i}}+a_{{i}}$ where $\alpha_{{i}}$ is its integer part whereas $a_{\rm{i}}$ is the fractional part satisfying $-1/2 <a_{{i}} \le 1/2$. Channel matrix (\ref{eq:Hi}) can be distinguished between two cases, namely integer Doppler shift and fractional Doppler shift \cite{AFDM_TWC}. 

\textit{i) Integer Doppler shifts:}
%With $k_i \in \mathbb{Z}$ for all $i \in [0,L-1] $, ${\bf{H}}_i[p,q]$ could be reformulated as
\begin{align}\label{eq:Hi_inte}
{\bf{H}}_{i}[p,q]=&e^{j\frac{2\pi}{N}(N c_{1}l_{i}^{2}-ql_{i}+N c_{2}(q^{2}-p^{2}))} \nonumber \\  &
\cdot \delta[ \langle p+loc_i \rangle_{N} - q],
\end{align}
where $loc_{i}\triangleq\langle 2N c_1 l_i - \alpha_{{i}} \rangle_{N}$.

\textit{ii) Fractional Doppler shifts:}
\begin{align}\label{eq:H_I_frac}
\vspace*{-2pt} %留空白，可自己调整
\left| {{{\bf{H}}_i}[p,q]} \right|  = 
\begin {cases}
\left| {\frac{{\sin \left( {N\theta } \right)}}{{N\sin \left( \theta  \right)}}} \right|  &{q \in {\mathbb{K}\left(p\right)}} \\
0    &{otherwise} \\
\end{cases},
%\text{,}
\vspace*{-2pt} %留空白，可自己调整
\end{align}
where ${\mathbb{K}\left(p\right)}=\left[ {{{\left\langle {p + lo{c_i} - {k_\nu }} \right\rangle }_N},{{\left\langle {p + lo{c_i} + {k_\nu }} \right\rangle }_N}} \right]$,$k_v\in \mathbb{N}$ denotes the considered Doppler spread range, and $\theta  \buildrel \Delta \over = \frac{\pi }{N}\left( {p - q + 2N{c_1}{l_i} - {k_i}} \right)$.

Observing (\ref{eq:Hi_inte}) and (\ref{eq:H_I_frac}), under both integer and fractional Doppler shift cases, the received signal is composed of multiple cyclically shifted copies of the transmitted signal. This phenomenon illustrates the cyclic-shift property of AFDM input-output relation. Jointly utilizing  this cyclic-shift property and impulse-like autocorrelation function of spreading sequences, we develop a correlation-based DAFT domain detector. The proposed CDD consists of equalization and despreading, which are detailed in the following.

For notational consistency, (\ref{eq:H_I_frac}) is used to represent the channel matrix under both integer and fractional Doppler shift cases in this paper. Specifically, the integer Doppler shift case corresponds to $k_{\nu}=0$. Then correlation-based equalization could be described as 
\begin{align}\label{eq:dce_fr}
{{\widehat{\bf{x}}}}_{\rm{d}}=\sum_{i=0}^{L-1}\frac{1}{{N}}h_{i}^{*}{\cdot} \sum_{k=-k_{\nu}}^{k_{\nu}}{\bf{\Pi}}^{\left(loc_{i}+k\right)}\left(\bm{\pi}_{i,k}\odot{\bf{{y}}}\right) ,
\end{align}
where ${\bf{\Pi }}\in {\mathbb{C}^{N \times N}}$ is the forward cyclic-shift matrix
\begin{equation}
{\bf{\Pi }}  = \left[ {\begin{array}{*{20}{c}}
	0&{...}&0&1\\
	1&{...}&0&0\\
	\vdots & \ddots & \ddots & \vdots \\
	0&{...}&1&0
	\end{array}} \right],
\end{equation}
and $\bm{\pi}_{i,k} \in \mathbb{C}^{{N}\times 1}$ is the phase vector whose elements is 
\begin{align}
\bm{\pi}_{i,k}[p] =& e^{-j\frac{2\pi}{{N}}\left(Nc_1l_i^2-\langle p+loc_{i}+k \rangle_{{N}}l_{i}+{N} c_{2}\left(\left( \langle p+loc_{i}+k \rangle_{{N}}\right)^{2}-p^{2} \right) \right)}
\nonumber \\  &
\cdot   \left(\frac{1-e^{-j2\pi(p-\langle p+loc_{i}+k \rangle_{{N}}-a_{i}+loc_{i})}}{1-e^{-j\frac{2\pi}{{N}}\left(p-\langle p+loc_{i}+k \rangle_{{N}}-a_{\rm{i}}+loc_{i} \right)}}\right)^{*}.
\end{align}
%DAFT-domain compensation described as (\ref{eq:dce_fr}) enables proposed anti-interference AFDM system to achieve a performance equivalent to that of maximal ratio combining (MRC). Full diversity gain could thereby be achieved.

After obtaining ${{\widehat{\bf{x}}}}_{\rm{d}}$ by (\ref{eq:dce_fr}), the despreading can be performed by
\begin{equation}\label{eq:DesDDCD}
{\widehat {\bf{c}}_{\rm{d}}} = {\left( {{\bf{d}}_{\rm{s}}^{\rm{H}}{\widehat {\bf{X}}}_{\rm{d}}} \right)^{\rm{T}}},
\end{equation}
where ${\bf{d}}_{\rm{s}}$ is the vector of spreading sequence under constellation mapping, ${{\widehat {\bf{X}}}_{\rm{d}}}$ is the reshaped matrix with dimensions determined by the length of the spreading sequence ${\bf{d}}_{\rm{s}}$. {{Subsequently, constellation-based demodulation can be applied to ${\widehat {\bf{c}}_{\rm{d}}}$.}}

It is worth noting that the despreading in (\ref{eq:DesDDCD}) also plays a key role in suppressing inter-path interference. For clarity, we illustrate this principle using a two-path ($L=2$) channel example, which can be easily extended to the case where $L\ge 3$. In the case of $L=2$, the noiseless received signal after equalization in (\ref{eq:dce_fr}) can be expressed as
\begin{align}\label{eq:Le2}
{\widehat {\bf{x}}_{\rm{d}}}=&\frac{{{{\left| {{h_1}} \right|}^2} + {{\left| {{h_2}} \right|}^2}}}{{{N^2}}}{\bf{x}} + \frac{{{h_1}^*{h_2}}}{{{N^2}}}{{\bf{\Pi }}^{lo{c_1} - lo{c_2}}}\left( {{\bm{\pi} _\Delta }^* \odot {\bf{x}}} \right)  \nonumber \\  
& + \frac{{{h_1}{h_2}^*}}{{{N^2}}}{{\bf{\Pi }}^{lo{c_2} - lo{c_1}}}\left( {{\bm{\pi} _\Delta } \odot {\bf{x}}} \right),
\end{align}				
%\begin{equation}\label{eq:Le2}
%{\widehat {\bf{x}}_{\rm{d}}} = \frac{{{{\left| {{h_1}} \right|}^2} + {{\left| {{h_2}} \right|}^2}}}{{{N^2}}}{\bf{x}} + \frac{{{h_1}^*{h_2}}}{{{N^2}}}{{\bf{\Pi }}^{lo{c_1} - lo{c_2}}}\left( {{\bm{\pi} _\Delta }^* \odot {\bf{x}}} \right) + \frac{{{h_1}{h_2}^*}}{{{N^2}}}{{\bf{\Pi }}^{lo{c_2} - lo{c_1}}}\left( {{\bm{\pi} _\Delta } \odot {\bf{x}}} \right),
%\end{equation}
where ${\bm{\pi} _\Delta }={\bm{\pi} _{1} }{\bm{\pi} _{2} }^*$. The terms $\frac{{{h_1}^*{h_2}}}{{{N^2}}}{{\bf{\Pi }}^{lo{c_1} - lo{c_2}}}\left( {{\bm{\pi} _\Delta }^* \odot {\bf{x}}} \right)$ and $\frac{{{h_1}{h_2}^*}}{{{N^2}}}{{\bf{\Pi }}^{lo{c_2} - lo{c_1}}}\left( {{\bm{\pi} _\Delta } \odot {\bf{x}}} \right)$ imply inter-path interference which are essentially cyclically shifted replicas of ${\bf{x}}$. The magnitudes of these terms after despreading are determined by the autocorrelation function of the spreading sequence. Since the autocorrelation sidelobes of the spreading sequence satisfy $R_d\left[k\right] \ll 1$ for $k\ne 0$, the resulting inter-path interference is effectively suppressed after despreading, while the original transmitted symbols are preserved. 

\subsection{Computational Complexity Analyses}
The computational complexity of CDD, as represented by (\ref{eq:dce_fr}) and (\ref{eq:DesDDCD}), is evaluated to be ${\mathcal {O}}\left( {k_vL{N}} \right)$. As contrast, the MMSE-based symbol detection method requires ${\mathcal {O}}\left( {N}^{3} \right)$ flops \cite{AFDM_TWC}.

\section{Numerical Simulations and Analyses}

In this section, we present the numerical results of the proposed anti-interference AFDM system. The simulation parameters are established based on 5G New Radio (NR) standard \cite{3GPP} and are summarized in Table \ref{tab:parameter_simu}. In the legends of this section, {{`TI' refers to tone interference,}} `STI' refers to single-tone interference, `SWI' refers to sweeping interference, `BBI' refers to broadband interference, and `NBI' refers to narrowband interference.
\begin{table}
	\centering
	%%\caption{****}
	\caption{Simulations Parameters}\label{tab:parameter_simu}
	\renewcommand\arraystretch{1.3}
	\begin{threeparttable}
		{
			\begin{tabular}{|c|l|c|}
				\hline
				\textbf{Symbol}                       &\textbf{Parameter}                       & \textbf{Value}     \\ \hline
				$f_{\rm{c}}$ &  Carrier frequency                      & $24$ GHz     \\ \hline
				${B}$ &  Bandwidth                       & $122.88$ MHz  \\ \hline
				$\Delta _f$ &  Subcarrier spacing                       & $15$ kHz  \\ \hline
				$N$ &  Number of subcarriers                      & $992$     \\ \hline
				$N_{\rm{cp}}$ &  Number of chirp-periodic prefix                     & $69$     \\ \hline
				$N_{\rm{m}}$ &  Modulation order of constellation diagram                      & $4$     \\ \hline	
				$N_{\rm{p}}$ &  Number of bits per packet                      & $544$     \\ \hline								
			\end{tabular}
		}
	\end{threeparttable}
\end{table}

Firstly, the accuracies of the impact analyses are verified through relative errors between the closed-form expressions in Section III and the results obtained by directly applying the DAFT. Considering the approximations inherent in the analyses of both STI and SWI, we conduct Monte Carlo simulations with $10^6$ trials to validate the analytical results under these two types of interference. In each trial, all interference signal parameters (carrier frequency, initial phase, and {{frequency}} modulation slope) are randomly generated from uniform distributions.
The statistical results are summarized in Table  \ref{tab:RelativeError}. Relative errors of STI and SWI impact, consistently below $-80$ dB, sufficiently validate the {{accuracy}} of the impact analyses.
\begin{table}
	\centering
	%%\caption{****}
	\caption{Statistical Results of Interference Impacts Analyses}\label{tab:RelativeError}
	\renewcommand\arraystretch{1.3}
	\begin{threeparttable}
		{
			\begin{tabular}{|c|l|c|}
				\hline
				\textbf{Item}                       &\textbf{Value}                           \\ \hline
				Relative error of STI impact in the DAFT domain &   $-117.02$ dB                       \\ \hline
				Relative error of SWI impact in the DAFT domain &    $-85.66$ dB                         \\ \hline
			\end{tabular}
			
		}
	\end{threeparttable}
\end{table}

Then, we show the packet throughput performance of the proposed anti-interference AFDM system. As \cite{AFDM_TWC}, we consider a channel with $L=3$ paths. The maximum relative velocity is $1000$ km/h \cite{IMT2021White}, {{which corresponds to a maximum Doppler shift of $22.22$ kHz,}} and the signal noise ratio (SNR) is set to $-10$ dB. Three paths have different delays which can be given by $l=[0,4,8]$. We employ the Reed–Solomon code RS(31,17) as ECC.
Figure \ref{fg:SAT_AFDM_PTSim} and Fig. \ref{fg:NoSAT_AFDM_PTSim} reveal the packet throughput performance of anti-interference AFDM system under stationary and non-stationary interference impact, respectively. In these two simulations, AFDM system with the same framework as shown in Fig. \ref{fg:FRAMEWORK} but with fixed $N_{\rm{d}}=16$ is used as benchmark. Meanwhile, `AFDM-F', `AFDM-A', and `PT' refer to AFDM system with fixed parameters, the proposed anti-interference AFDM system, and packet throughput, respectively. 

Figure \ref{fg:SAT_AFDM_PTSim} {{shows}} the packet throughput performances of AFDM-F and AFDM-A {{versus}} interference signal ratio (ISR) under stationary interference impact. TI, SWI with frequency modulation slope differing from that of AFDM subcarrier, BBI and NBI are considered in this simulation. Under different types of interference, AFDM-A outperforms AFDM-F in terms of packet throughput performance.
These two systems exhibit comparable performance at an ISR of $8$ dB, as the parameters of AFDM-F are close to the optimal values under this interference level. 
When ISR exceeds $8$ dB, the packet throughput performance of AFDM-F degrades rapidly. This is because the interference strategies parameters of AFDM-F are not sufficient to suppress high-power interference.
In contrast, the packet throughput of AFDM-A decreases more gradually with increasing ISR, owing to its adaptive parameters that enable it to handle interference with varying power levels.

Figure \ref{fg:NoSAT_AFDM_PTSim} {{shows}} the packet throughput performances of AFDM-F and AFDM-A {{versus}} ISR under non-stationary interference impact. This simulation considers SWI with a frequency modulation slope aligned with that of the AFDM subcarriers. 
The packet throughput of AFDM-F degrades from $1750.18$ packets/s to $1749.64$ packets/s as ISR increases. Whereas the packet throughput of AFDM-F degrades from $2507.12$ packets/s to $2492.64$ packets/s as ISR increases. This is because the errors caused by the non-stationary interference impact, concentrating on a specific subcarrier, could be corrected by ECC. As a result, the effect of interference power on packet throughput is relatively limited. Additionally, appended graphics in Fig. \ref{fg:SAT_AFDM_PTSim} and Fig. \ref{fg:NoSAT_AFDM_PTSim} validate the derivation of packet throughput. 
\begin{figure}
	\centering
	\includegraphics[width=2.2in]{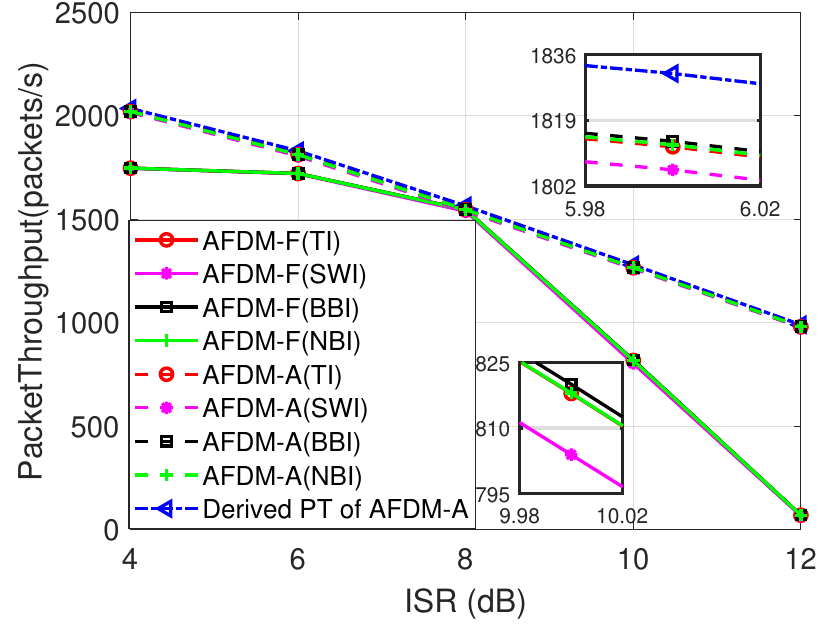}
	\vspace*{-8pt} %留空白，可自己调整
	\caption{{{Packet throughput versus ISR under stationary interference impact.}}
		\label{fg:SAT_AFDM_PTSim}} 
	\vspace*{-10pt} %留空白，可自己调整
\end{figure}

\begin{figure}
	\centering
	\includegraphics[width=2.2in]{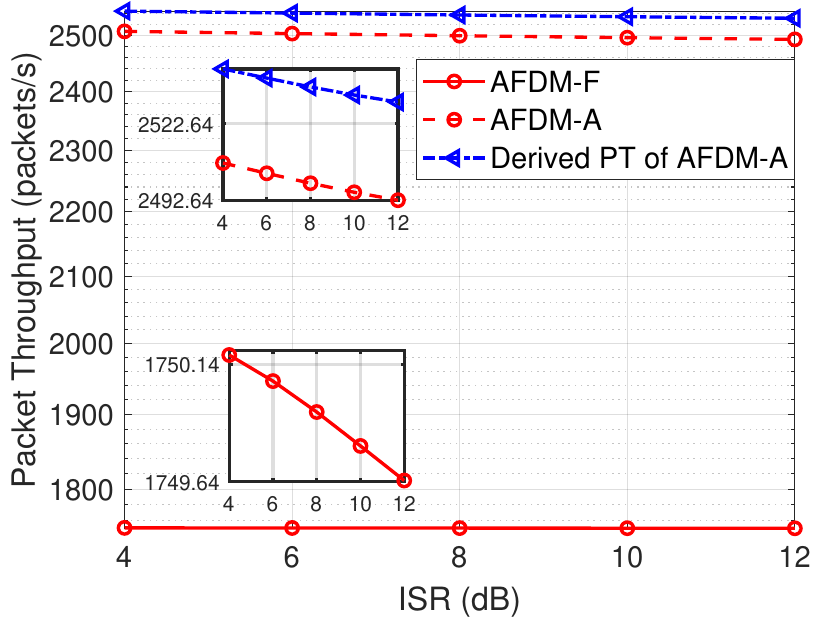}
	\vspace*{-8pt} %留空白，可自己调整
	\caption{Packet throughput versus ISR under non-stationary interference impact.  
		\label{fg:NoSAT_AFDM_PTSim}} 
	\vspace*{-12pt} %留空白，可自己调整
\end{figure}

As shown in Fig. \ref{fg:AFDMvsOTFSvsOFDMSim}, we compare the packet throughput performances of OTFS, OFDM, and AFDM systems under BBI. Similar to Fig. \ref{fg:SAT_AFDM_PTSim}, OTFS, OFDM, and AFDM systems with the same framework as shown in Fig. \ref{fg:FRAMEWORK} but with fixed parameters are considered. `OTFS-F' and `OFDM-F' refer to OTFS and OFDM systems with fixed parameters, respectively. When ISR exceeds $8$ dB, the packet throughput of both AFDM-F and OTFS-F degrades from $1543.46$ packets/s to $68.77$ packets/s. Whereas the packet throughput of OFDM-F degrades from $1471.51$ packets/s to $42.97$ packets/s when ISR exceeds $8$ dB. The inferior packet throughput performance of OFDM-F is due to limited diversity gain of OFDM. The packet throughput of AFDM-A degrades from $1547.23$ packets/s to $981.34$ packets/s when ISR exceeds $8$ dB. This is owing to the capability of AFDM-A to suppress interference with varying power levels.
\begin{figure}
	\centering
	\includegraphics[width=2.2in]{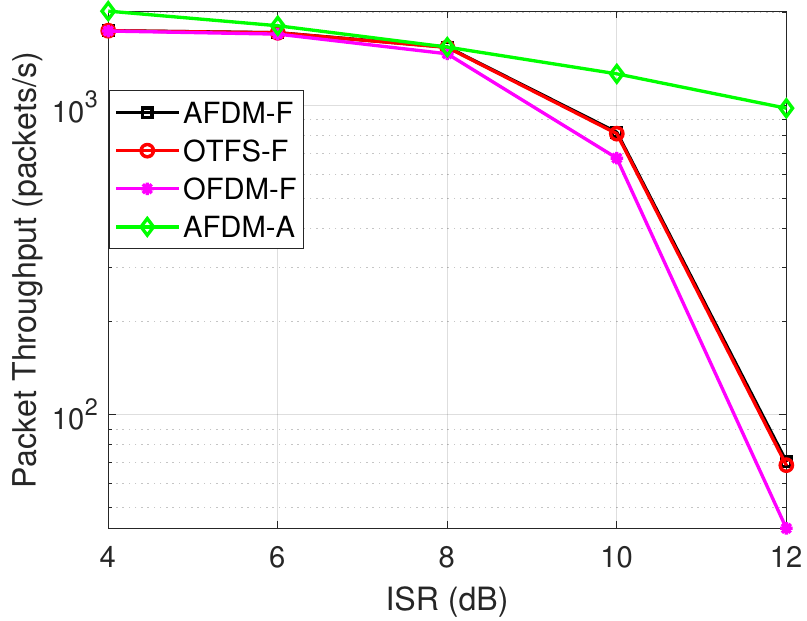}
	\vspace*{-8pt} %留空白，可自己调整
	\caption{ Packet throughput versus ISR for OTFS, OFDM and AFDM.
		\label{fg:AFDMvsOTFSvsOFDMSim}} 
	\vspace*{-10pt} %留空白，可自己调整
\end{figure}

\begin{figure}
	\centering
	\includegraphics[width=2.2in]{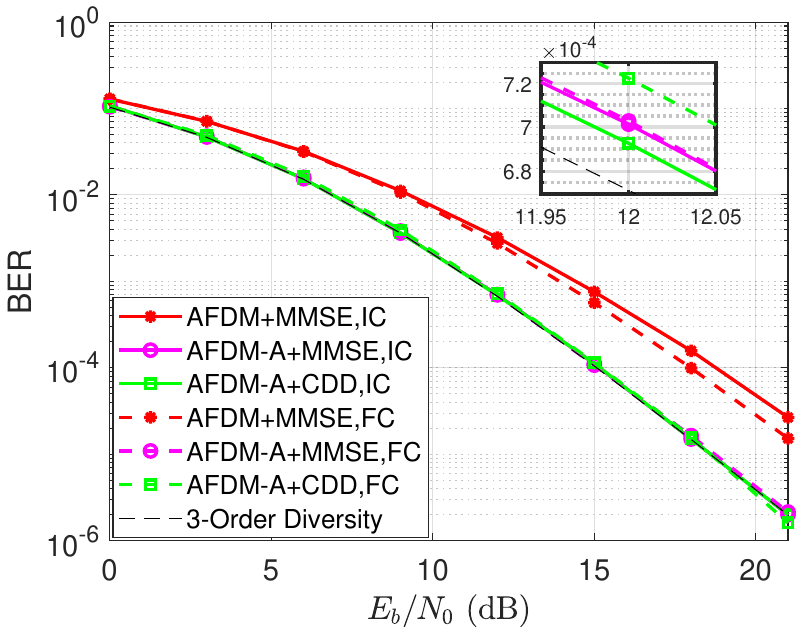}
	\vspace*{-8pt} %留空白，可自己调整
	\caption{ BER performance comparison between AFDM and AFDM-A systems using different detectors.
		\label{fg:BERvsEbN0}} 
	\vspace*{-15pt} %留空白，可自己调整
\end{figure}

Finally, we conduct a simulation to evaluate the BER performance of the proposed CDD method. Both integer Doppler and fractional Doppler cases are considered in this simulation, denoted as `IC' and `FC', respectively.
{{The parameter settings for the fractional-Doppler case are kept identical to that used in the packet throughput simulations, whereas for the integer Doppler case, the normalized Doppler values $k_{i}$ are set to $\left[2,1,0\right]$. In addition, $k_{v} = 7$ is used for CDD.}}
Figure \ref{fg:BERvsEbN0} shows BER of AFDM and AFDM-A versus energy per bit to noise power spectral density ratio ($E_b/N_0$). Note that ``3-Order Diversity'' represents the optimal diversity order of the $3$-path channel.
We observe that AFDM-A outperforms AFDM in both IC and FC cases. This is because MMSE equalization does not introduce inter-path interference for AFDM-A, as analyzed in Section V. 
In both IC and FC cases, the BER of the AFDM-A system using CDD and MMSE detectors closely approaches the maximum diversity of the $3$-path channel. This confirms that CDD is capable of achieving the maximum diversity order with linear computational complexity. 

\section{Conclusion}

This paper proposed an anti-interference AFDM system to ensure reliability and resource efficiency under malicious interference in high-mobility scenarios.
We derived closed-form expressions for interferences in the DAFT domain to analyze interference impact in the DAFT domain. 
Based on the analyses, we derive the analytical relationship between packet throughput and system {{parameters}}, and subsequently design a parameter optimization algorithm.
Finally, a linear-complexity symbol detection method capable of achieving full diversity gain was presented for the proposed anti-interference AFDM system.
Numerical results validated the accuracy of the derived closed-form expressions and verified that the proposed anti-interference AFDM system could achieve high packet throughput under interference in high-mobility scenarios.

\section*{APPENDIX}

\appendices

\section{Proof of Proposition 1}\label{proof_Corollary1}
\vspace*{-1pt} %留\ref空白，可自己调整
\small
{{First, we consider the case $N_{\rm{i}}=1$, corresponding to single-tone interference.}} Based on Lemma 1, we have
\begin{equation}
\vspace*{-2pt} %留空白，可自己调整
{
	\label{eq:Lf}
	L_{\rm{f}}^{\rm{A}}\left( u \right) = {e^{j\frac{a}{{2b}}{u^2}}}\sqrt{\frac{{{P_{\rm{i}}}}}{{ d }}}{e^{ - j\left(\frac{\pi }{4}-\theta_{\rm{i}}\right)}}{e^{j\frac{{{p^2}\left( {{f_{\rm{m}}},u} \right)}}{{2bd}}}}g\left( {\frac{{p\left( {{f_{\rm{m}}},u} \right)}}{{dT}}} \right)\text{,} }
\vspace*{-2pt} %留空白，可自己调整
\end{equation}
where $p\left( {{f_{\rm{m}}},u} \right) = 2\pi b {f_{\rm{m}}} - u$. By utilizing the relationship between Fourier series and the Fourier transform, ${F_{\rm{q}}}\left( {\frac{u}{b}} \right)$ can be given by   
\begin{equation}
\vspace*{-2pt} %留空白，可自己调整
{
	\label{eq:Fq}
	{F_{\rm{q}}}\left( {\frac{u}{b}} \right) = 2\pi b{f_{\rm{s}}}\sum\limits_{k =  - \infty }^\infty  {\delta \left( {u - 2\pi kb{f_{\rm{s}}}} \right)}\text{.} }
\vspace*{-2pt} %留空白，可自己调整
\end{equation}
Substituting (\ref{eq:Lf}) and (\ref{eq:Fq}) into $\left[ {L_{\rm{f}}^{\rm{A}}\left( u \right){e^{j\frac{a}{{2b}}{u^2}}}} \right] * {F_{\rm{q}}}\left( {u/b} \right)$, we get (\ref{eq:STEP1Lemma2}) on the next page. 
\begin{figure*}[htbp]
	\vspace*{-10pt} %留空白，可自己调整
	\small
	\begin{flalign}\label{eq:STEP1Lemma2}
	\left[ {L_{\rm{t}}^{\rm{A}}\left( u \right){e^{j\frac{a}{{2b}}{u^2}}}} \right] * {F_{\rm{q}}}\left( {u/b} \right) = 2\pi b{f_{\rm{s}}}\sqrt{\frac{{{P_{\rm{i}}}}}{{ d }}}{e^{ - j\left(\frac{\pi }{4}-\theta_{\rm{i}}\right)}}
	\sum\limits_{k =  - \infty }^\infty  {{e^{\frac{{j{f_{\rm{s}}}^2}}{{2bd}}{p^2}\left( {{f_{\rm{m}}}/{f_{\rm{s}}},u/{f_{\rm{s}}} - 2\pi kb} \right)}}g\left( {\frac{{p\left( {{f_{\rm{m}}}/{f_{\rm{s}}},u/{f_{\rm{s}}} - 2\pi kb} \right){f_{\rm{s}}}}}{{dT}}} \right)}. 
	\end{flalign}\normalsize
	%	\hrule
	\vspace*{-10pt} %留空白，可自己调整
\end{figure*}
Based on (\ref{eq:REC}), $k$ is constrained by
\begin{equation}
\vspace*{-2pt} %留空白，可自己调整
{
	\label{eq:kRange}
	\left\lceil {\frac{u}{{2\pi b{f_{\rm{s}}}}} - \frac{{{f_{\rm{m}}}}}{{{f_{\rm{s}}}}}} \right\rceil  \le k \le \left\lceil {\frac{u}{{2\pi b{f_{\rm{s}}}}} - \frac{{{f_{\rm{m}}}}}{{{f_{\rm{s}}}}}} \right\rceil  + 2N{c_1} - 1 \text{.} }
\vspace*{-2pt} %留空白，可自己调整
\end{equation}
Let ${{{f_{\rm{m}}}}}/{{{f_{\rm{s}}}}} = {a_{\rm{f}}} + {\alpha_{\rm{f}}}$, where $a_{\rm{f}}$ is the integer part whereas $\alpha_{\rm{f}}$ is the fractional part satisfying $0 \le {\alpha_{\rm{f}}} < 1$. And we define $\alpha  = N{\alpha_{\rm{f}}}$. Since $N$ is generally taken as a large value in AFDM, the fractional part of $\alpha$ can be neglected in its impact on $\frac{{{f_{\rm{m}}}}}{{{f_{\rm{s}}}}}$, allowing $\alpha$ to be treated as an integer satisfying $0 \le \alpha  < N$.
%Then we have ${{{f_{\rm{m}}}}}/{{{f_{\rm{s}}}}} = {a_{\rm{f}}} + \alpha/N$. 
Considering that $u=m\Delta u$, (\ref{eq:kRange}) can be rewritten as
\begin{align}\label{eq:kRange2}
\vspace*{-2pt} %留空白，可自己调整
k \in \left\{ {\begin{array}{*{20}{c}}
	{\left[ {1 - a_{\rm{f}},2N{c_1} - a_{\rm{f}}} \right],\alpha  < m \le N + \alpha }\\
	{\left[ { - a_{\rm{f}},2N{c_1} - a_{\rm{f}} - 1} \right],\alpha  - N < m \le \alpha }
	\end{array}} \right.
\text{.}
\vspace*{-2pt} %留空白，可自己调整
\end{align}
When $\alpha  < m \le N + \alpha$, let ${{k^{'}} = k - 1 + a_{\rm{f}}}$, we have
\begin{align}\label{eq:p1}
\vspace*{-2pt} %留空白，可自己调整
p\left( {{f_{\rm{m}}}/{f_{\rm{s}}},u/{f_{\rm{s}}} - 2\pi kb} \right)&{=} 2\pi b\left( {{k^{'}} + \frac{{\alpha  - m + N}}{N}} \right) \nonumber\\
&{=}2\pi b\left( {{k^{'}} + \frac{{{{\left\langle {\alpha  - m} \right\rangle }_{N}}}}{N}} \right)
\text{.}
\vspace*{-2pt} %留空白，可自己调整
\end{align}
Whereas $\alpha  - N < m \le \alpha$, let ${{k^{'}} = k+ a_{\rm{f}}}$, we have
\begin{align}\label{eq:p2}
\vspace*{-2pt} %留空白，可自己调整
p\left( {{f_{\rm{m}}}/{f_{\rm{s}}},u/{f_{\rm{s}}} - 2\pi kb} \right)&{=} 2\pi b\left( {{k^{'}} + \frac{{\alpha  - m}}{N}} \right) \nonumber\\
&{=}2\pi b\left( {{k^{'}} + \frac{{{{\left\langle {\alpha  - m} \right\rangle }_{N}}}}{N}} \right)
\text{.}
\vspace*{-2pt} %留空白，可自己调整
\end{align}
Then the interference in the DAFT domain can be calculated through
\begin{equation}
\vspace*{-2pt} %留空白，可自己调整
{
	\label{eq:STI_DAFT2}
	J_{\rm{t}}^{\rm{A}}\left( m \right) = \sqrt {\frac{{2\pi b{}}}{N}}  L_{\rm{st}}^{\rm{A}}\left( u \right){|_{u = m\Delta u}} \text{,} }%{e^{j{\theta _{\rm{i}}}}}
\vspace*{-2pt} %留空白，可自己调整
\end{equation}
where $\Delta u = {2\pi  b {f_{\rm{s}}}}/{N}$, $\sqrt {{2\pi b}/{N}}$ is the normalization factor which ensures the conjugate invertibility of the transform. 
By defining
\begin{equation}
\vspace*{-2pt} %留空白，可自己调整
{
	\label{eq:Lnm}
	L\left( {n,m} \right) = \sum\limits_{k = 0}^{\left| n \right| - 1} {{e^{j\frac{\pi }{{Nn}}{{\left( {Nk + m} \right)}^2}}}} \text{,} }
\vspace*{-2pt} %留空白，可自己调整
\end{equation}
(\ref{eq:STI_DAFT2}) can be reformulated in closed-form as
\begin{equation}
\vspace*{-2pt} %留空白，可自己调整
{
	\label{eq:STI_DAFT3}
	J_{\rm{t}}^{\rm{A}}\left( m \right) = \sqrt {\frac{{{P_{\rm{i}}}}}{{2N{c_1}}}} {e^{ - j\left( {2\pi {c_2}{m^2} + \frac{\pi }{4} - {\theta _{\rm{i}}}} \right)}}L\left( {2N{c_1},{{\left\langle {\alpha  - m} \right\rangle }_{N}}} \right)\text{.} }
\vspace*{-0pt} %留空白，可自己调整
\end{equation}

To further analyze (\ref{eq:STI_DAFT3}), we derive the magnitude of $L\left( {n,m} \right)$ as
\begin{align}\label{eq:L2}
\vspace*{-2pt} %留空白，可自己调整
&{{\left| {L\left( {n,m} \right)} \right|^2}} = \sum\limits_{{k_1} = 0}^{\left| n \right| - 1} {{e^{j\frac{\pi }{{Nn}}{{\left( {N{k_1} + m} \right)}^2}}}} \sum\limits_{{k_2} = 0}^{\left| n \right| - 1} {{e^{ - j\frac{\pi }{{Nn}}{{\left( {N{k_2} + m} \right)}^2}}}} \nonumber\\
&{\overset{(a)}{=}} \left| n \right| + \sum\limits_{{k_1} = 0}^{\left| n \right| - 1} {{e^{j\frac{\pi }{{Nn}}{{\left( {N{k_1} + m} \right)}^2}}}} \sum\limits_{{k_2} = 1}^{\left| n \right| - 1} {{e^{ - j\frac{\pi }{{Nn}}{{\left( {N{k_1} + N{k_2} + m} \right)}^2}}}}  \nonumber\\
%&{=} \left| n \right| + \sum\limits_{{k_2} = 1}^{\left| n \right| - 1} {{e^{j\frac{\pi }{{Nn}}\left( {{N^2}{k_2}^2 + 2mN{k_2}} \right)}}} \sum\limits_{{k_1} = 0}^{\left| n \right| - 1} {{e^{ - j\frac{\pi }{{Nn}}\left( {2{N^2}{k_1}{k_2}} \right)}}}  \nonumber\\
&{=} \left| n \right|
\text{,}
\vspace*{-2pt} %留空白，可自己调整
\end{align}
where $(a)$ follows from that sequence ${{e^{ - j\frac{\pi }{{Nn}}{{\left( {N{k_2} + m} \right)}^2}}}}$ is periodic with a period of $\left| n \right|$. Then we can get $ {{L\left( {n,m} \right)} } = \sqrt {\left| n \right|} {e^{j{\theta _{\rm{L}}}}}$ where ${\theta _{\rm{L}}}$ is the phase of $L\left( {n,m} \right)$.

Accordingly, when $N_{\rm{i}}=1$, $J_{\rm{t}}^{\rm{A}}\left( m \right)$ is a complex number with amplitude $\sqrt {{P_{\rm{i}}}} $ and phase shift $\left({\theta _{\rm{i}}} + {\theta _{\rm{L}}} - 2\pi {c_2}{m^2} - {\pi }/{4}\right) $. Because the initial phase $\theta_{\rm{i}}$, is independent of $\left({\theta _{\rm{L}}} - 2\pi {c_2}{m^2} - {\pi }/{4}\right) $ and is uniformly distributed within $\left[ {-\pi,\pi } \right]$, the phase of $J_{\rm{t}}^{\rm{A}}\left( m \right)$ could be considered to follow a uniform distribution within $\left[ {-\pi,\pi } \right]$.

{{Then we consider the case $N_{\rm{i}}>1$, corresponding to multiple-tone interferenceinterference. The resulting DAFT domain interference can be regarded as the superposition of the DAFT domain contributions of multiple independent single-tone components. In the signal model given in (\ref{eq:STI}), the phase of each tone is independent and uniformly distributed within $\left[ {-\pi,\pi } \right]$. Therefore, $J_{\rm{t}}^{\rm{A}}\left( m \right)$ for $N_{\rm{i}}>1$ is the sum of $N_{\rm{i}}$ complex random variables with identical amplitude and mutually independent random phases. Based on the central limit theorem, $J_{\rm{t}}^{\rm{A}}\left( m \right)$ converges in distribution to a complex centered Gaussian distribution for sufficiently large $N_{\rm{i}}$. The mean and variance are given by ${\mathbb{E}}\left\{ {J_{\rm{t}}^{\rm{A}}\left( m \right)} \right\} = 0,{\mathbb{V}}\left\{ {J_{\rm{t}}^{\rm{A}}\left( m \right)} \right\} = P_{\rm{i}}$, i.e., $J_{\rm{t}}^{\rm{A}}\left( m \right) \sim \mathcal {CN}\left( {0\text{,} P_{\rm{i}}} \right)$.}}

Proposition 1 is proved.

\section{Proof of Proposition 2}\label{proof_Corollary2}
\vspace*{-1pt} %留\ref空白，可自己调整
\small
Based on the AFT convolution theorem, ${L_{\rm{sw}}}\left( u \right)$ can be given by
\begin{equation}
\vspace*{-2pt} %留空白，可自己调整
{
	\label{eq:CT_SWI}
	{L_{\rm{sw}}}\left( u \right) = \frac{{{e^{ - j\frac{a}{{2b}}{u^2}}}}}{{2\pi b}}\left\{ {\left[ {L_{\rm{w}}^{\rm{A}}\left( u \right){e^{j\frac{a}{{2b}}{u^2}}}} \right] * {F_{\rm{q}}}\left( {u/b} \right)} \right\}\text{,} }
\vspace*{-2pt} %留空白，可自己调整
\end{equation}
where ${L_{\rm{w}}^{\rm{A}}\left( u \right)}$ is the AFT of $f_{\rm{sw}}\left( t \right)$. By employing the stationary phase principle, we have a closed-from of ${L_{\rm{w}}^{\rm{A}}\left( u \right)}$ as (\ref{eq:LCT_SWI}) on the next page, where $C.1$ denotes the constraint that ${{\varphi _{\rm{i}}} = \frac{d}{{2\pi b}}}$ and ${{f_{\rm{m}}} = \frac{u}{{2\pi b}}}$.
\begin{figure*}[htbp]
	\vspace*{-10pt} %留空白，可自己调整
	\small
	%\begin{flalign}\label{eq:LCT_SWI}
	\begin{flalign}\label{eq:LCT_SWI}
	\vspace*{-2pt} %留空白，可自己调整
	{L_{\rm{w}}^{\rm{A}}}\left( u \right)  =
	\begin {cases}
	{\sqrt {\frac{j}{{2\pi b{\varphi _{\rm{i}}} - d}}} {e^{ - \frac{j}{{2b}}\left( {a{u^2} + \frac{{{{\left( {2\pi b{f_{\rm{m}}} - u} \right)}^2}}}{{2\pi b{\varphi _{\rm{i}}} - d}}} \right)}}g\left( {\frac{{p\left( {{f_{\rm{m}}},u} \right)}}{{\left( {d - 2\pi b{\varphi _{\rm{i}}}} \right)T}}} \right)}\quad &{{\varphi _{\rm{i}}}} { \ne \frac{d}{{2\pi b}}}\\
	{\frac{T}{{\sqrt {2\pi b} }}{e^{ - j\frac{a}{{2b}}{u^2}}}} \quad &{C.1} \\
	{\frac{{ - jb}}{{\sqrt {2\pi b} }}{e^{ - j\frac{a}{{2b}}{u^2}}}\frac{{{e^{j\left( {2\pi b{f_{\rm{m}}} - u} \right)\frac{T}{b}}} - 1}}{{p\left( {{f_{\rm{m}}},u} \right)}}} \quad &{otherwise} \\
	\end{cases}.
	%\text{,}
	\vspace*{-2pt} %留空白，可自己调整
	\end{flalign}
	%\end{flalign}\normalsize
%	\hrule
	\vspace*{-10pt} %留空白，可自己调整
\end{figure*}
Similar to Appendix \ref{proof_Corollary1}, we have ${{{f_{\rm{m}}}}}/{{{f_{\rm{s}}}}} = a + {\alpha }/{N}$ where $\alpha={N{{\left\langle {{{{f_{\rm{m}}}}}/{{{f_{\rm{s}}}}}} \right\rangle }_1}}$ is treated as an integer satisfying $0 \le \alpha  < N$. Based on (\ref{eq:LCT_SWI}) and (\ref{eq:Fq}), $\left[ {{L_{\rm{w}}}\left( u \right){e^{j\frac{a}{{2b}}{u^2}}}} \right] * {F_{\rm{q}}}\left( {\frac{u}{b}} \right){|_{u = m\Delta u}}$ in (\ref{eq:SWI_DAFT2}) could be written as (\ref{eq:STEP2Lemma3}) on the next page, where $L\left(n,m\right)$ is defined as (\ref{eq:Lnm}).
\begin{figure*}[htbp]
	\vspace*{-10pt} %留空白，可自己调整
	\small
	%\begin{flalign}\label{eq:LCT_SWI}
	\begin{flalign}\label{eq:STEP2Lemma3}
	\vspace*{-2pt} %留空白，可自己调整
	\left[ {{L_{\rm{w}}}\left( u \right){e^{j\frac{a}{{2b}}{u^2}}}} \right] * {F_{\rm{q}}}\left( {\frac{u}{b}} \right){|_{u = m\Delta u}} =
	\begin {cases}
	{\sqrt {\frac{{j2\pi bN}}{{{N_{\rm{s}}} - 2N{c_1}}}} L\left( {2N{c_1} - {N_{\rm{s}}},{{\left\langle {{N_{\rm{s}}} - m} \right\rangle }_N}} \right)} &{{\varphi _{\rm{i}}}} { \ne \frac{d}{{2\pi b}}}\\
	{N\sqrt {2\pi b} \delta \left( {m - \alpha } \right)} \quad &{{\varphi _{\rm{i}}}} { = \frac{d}{{2\pi b}}} \\
	\end{cases}.
	%\text{,}
	\vspace*{-2pt} %留空白，可自己调整
	\end{flalign}
	%\end{flalign}\normalsize
	\vspace*{-10pt} %留空白，可自己调整
\end{figure*}

Thus, we can reformulate (\ref{eq:SWI_DAFT2}) in closed-form as (\ref{eq:SWI_DAFT3}).
\begin{figure*}[htbp]
	\vspace*{-10pt} %留空白，可自己调整
	\small
	%\begin{flalign}\label{eq:LCT_SWI}
	\begin{flalign}\label{eq:SWI_DAFT3}
	\vspace*{-2pt} %留空白，可自己调整
	J_{{\rm{sw}}}^{\rm{A}}\left( m \right)  =
	\begin {cases}
	{\sqrt {\frac{{j{P_{\rm{i}}}}}{{{N_{\rm{s}}} - 2N{c_1}}}} {e^{ - j\left( {2\pi {c_2}{m^2} - {\theta _{\rm{i}}}} \right)}}L\left( {2N{c_1} - {N_{\rm{s}}},{{\left\langle {{N_{\rm{s}}} - m} \right\rangle }_N}} \right)} &{{\varphi _{\rm{i}}}} { \ne \frac{d}{{2\pi b}}}\\
	{\sqrt {{NP_{\rm{i}}}} {e^{ - j\left(2\pi{c_2}{\alpha ^2}-{\theta _{\rm{i}}}\right)}}\delta \left( {m - \alpha } \right)} \quad &{{\varphi _{\rm{i}}}} { = \frac{d}{{2\pi b}}} \\
	\end{cases}.
	%\text{,}
	\vspace*{-2pt} %留空白，可自己调整
	\end{flalign}
	%\end{flalign}\normalsize
	\hrule
	\vspace*{-10pt} %留空白，可自己调整
\end{figure*}
When ${{\varphi _{\rm{i}}} > \frac{d}{{2\pi b}}}$, $J_{\rm{sw}}^{\rm{A}}\left( m \right)$ is with amplitude $\sqrt {{P_{\rm{i}}}} $ and phase shift $\left({\theta _{\rm{i}}} + {\theta _{\rm{L}}} - 2\pi {c_2}{m^2} + {\pi }/{4}\right)$. Whereas ${{\varphi _{\rm{i}}} <\frac{d}{{2\pi b}}}$, $J_{\rm{sw}}^{\rm{A}}\left( m \right)$ is with amplitude $\sqrt {{P_{\rm{i}}}} $ and phase shift $\left({\theta _{\rm{i}}} + {\theta _{\rm{L}}} - 2\pi {c_2}{m^2} - {\pi }/{4}\right)$.
Moreover, phase of $J_{\rm{sw}}^{\rm{A}}\left( m \right)$ is ${\theta _{\rm{i}}}  - 2\pi {c_2}{m^2} $ when ${{\varphi _{\rm{i}}} = \frac{d}{{2\pi b}}}$. Since $\theta_{\rm{i}}$, is uniformly distributed within $\left[ {-\pi,\pi } \right]$, phase of $J_{\rm{sw}}^{\rm{A}}\left( m \right)$ could be considered to follow a uniform distribution within $\left[ {-\pi,\pi } \right]$.

Proposition 2 is proved.

\section{Proof of Proposition 4}\label{proof_Corollary4}
\vspace*{-20pt} %留\ref空白，可自己调整
\small
\begin{equation}
\vspace*{-2pt} %留空白，可自己调整
{
	\label{eq:E_first}
	\mathbb{E}\left\{ {J_{\rm{nb}}^{\rm{A}}\left( m \right)} \right\} =K\left( {m,{\theta _{\rm{i}}}} \right)\sum\limits_{n = 0}^{N - 1} {\mathbb{E}\left\{ {J_{\rm{nb}}^i\left( n \right)} \right\}} A\left( {m,n,{f_{\rm{d}}}} \right) \text{,} }
\vspace*{-2pt} %留空白，可自己调整
\end{equation}
where $K\left( {m,{\theta _{\rm{i}}}} \right) = \sqrt {{{{P_{\rm{i}}}}}/{N}} {e^{ - j\left( {2\pi {c_2}{m^2} + {\theta _{\rm{i}}}} \right)}}$, $A\left( {m,n,{f_{\rm{d}}}} \right) = {e^{ - j2\pi \left( {\frac{{mn}}{N} + {c_1}{n^2} - {f_{\rm{d}}}n} \right)}}$. When $i=1$, the expectation term on the right-hand side of (\ref{eq:E_first}) can be reformulated as
\begin{equation}
\vspace*{-2pt} %留空白，可自己调整
{
	\label{eq:E_second1}
	{\mathbb{E}\left\{ {J_{\rm{nb}}^1\left( n \right)} \right\}} =\sum\limits_{k =  - \infty }^\infty  {\mathbb{E}\left\{ {z\left( {n - k} \right)} \right\}h\left( {{B_{\rm{i}}},k} \right) = 0}  \text{,} }
\vspace*{-2pt} %留空白，可自己调整
\end{equation}
where ${\mathbb{E}\left\{ {z\left( n \right)} \right\}}=0$. While ${\mathbb{E}\left\{ {J_{\rm{nb}}^2\left( n \right)} \right\}}$ could be given by
\begin{equation}
\vspace*{-2pt} %留空白，可自己调整
{
	\label{eq:E_second2}
	{\mathbb{E}\left\{ {J_{\rm{nb}}^2\left( n \right)} \right\}} =\sum\limits_{p =  0 }^\infty  {\mathbb{E}\left\{ {a\left( {p} \right)} \right\}g\left( n-pR_{\rm{u}} \right) = 0}  \text{,} }
\vspace*{-2pt} %留空白，可自己调整
\end{equation}
where ${\mathbb{E}\left\{ {a\left( p \right)} \right\}}=0$. Thus, $\mathbb{E}\left\{ {J_{\rm{nb}}^{\rm{A}}\left( m \right)} \right\}=0$.

Additionally,
\begin{align}\label{eq:PowerNBI}
\vspace*{-2pt} %留空白，可自己调整
{\left| {J_{\rm{nb}}^{\rm{A}}\left( m \right)} \right|^2} = \frac{{{P_{\rm{i}}}}}{N}&{\cdot}\sum\limits_{{n_1} = 0}^{N - 1} {\left\{ {{{\left| {J_{\rm{nb}}^i\left( {{n_1}} \right)} \right|}^2} + \sum\limits_{{n_2} \ne {n_1}} {C\left( {{n_1},{n_2}} \right)} } \right\}} \nonumber\\
= \frac{{{P_{\rm{i}}}}}{N}&{\cdot}\left(S_1+S_2\right)
\text{,}
\vspace*{-2pt} %留空白，可自己调整
\end{align}
where $C\left( {{n_1},{n_2}} \right) = J_{\rm{nb}}^i\left( {{n_1}} \right){\left[ {J_{\rm{nb}}^i\left( {{n_2}} \right)} \right]^*}{\Delta _a}\left( {{n_1},{n_2}} \right)$, ${\Delta _a}\left( {{n_1},{n_2}} \right) = {e^{j2\pi \left( {\frac{{m\left( {{n_2} - {n_1}} \right)}}{N} + {c_1}\left( {{n_2}^2 - {n_1}^2} \right) - {f_{\rm{d}}}\left( {{n_2} - {n_1}} \right)} \right)}}$,
\begin{equation}
\vspace*{-2pt} %留空白，可自己调整
{
	\label{eq:S1}
	{S_1} = \sum\limits_{{n_1} = 0}^{N - 1} {{{\left| {J_{\rm{nb}}^i\left( {{n_1}} \right)} \right|}^2}} \text{,} }
\vspace*{-2pt} %留空白，可自己调整
\end{equation}
and
\begin{equation}
\vspace*{-2pt} %留空白，可自己调整
{
	\label{eq:S2}
	{S_2} = \sum\limits_{{n_1} = 0}^{N - 1} {\sum\limits_{{n_2} \ne {n_1}} {C\left( {{n_1},{n_2}} \right)} }  \text{.} }
\vspace*{-2pt} %留空白，可自己调整
\end{equation}
The expectation of $S_2$ can be derived as
\begin{align}\label{eq:E_S2}
\vspace*{-2pt} %留空白，可自己调整
\mathbb{E}\left\{ {{S_2}} \right\} = &{\sum\limits_{{n_1} = 0}^{N - 1} }{{e^{ - j2\pi \left( {\frac{{m{n_1}}}{N} + {c_1}n_1^2 - {f_{\rm{d}}}{n_1}} \right)}}} \nonumber\\
&{\sum\limits_{{n_2} \ne {n_1}}} {{e^{j2\pi \left( {\frac{{m{n_2}}}{N} + {c_1}n_2^2 - {f_{\rm{d}}}{n_2}} \right)}}}\mathbb{E}\left\{ {\Delta _{nb}^i\left( {{n_1},{n_2}} \right)} \right\} \nonumber\\
\overset{(a)}{=}&{\sum\limits_{{n_1} = 0}^{N - 1} }{{e^{ - j2\pi \left( {\frac{{m{n_1}}}{N} + {c_1}n_1^2 - {f_{\rm{d}}}{n_1}} \right)}}} \nonumber\\
&{{\sum\limits_{\Delta n = 1}^{N - 1}} {{e^{j2\pi \left( {\frac{{m\left( {n_2} \right)}}{N} + {c_1}{{\left( {n_2} \right)}^2} - {f_{\rm{d}}}\left( {n_2} \right)} \right)}}} }{\cdot}\mathbb{E}\left\{ {\Delta _{nb}^i\left( {0,\Delta n} \right)} \right\}\nonumber\\
%=&{\sum\limits_{\Delta n = 1}^{N - 1}} {{e^{j2\pi \left( {\frac{{m\Delta n}}{N} + {c_1}\Delta {n^2} - {f_{\rm{d}}}\Delta n} \right)}}\mathbb{E}\left\{ {\Delta _{nb}^i\left( {0,\Delta n} \right)} \right\}}\nonumber\\
{=}&0
\text{,}
\vspace*{-2pt} %留空白，可自己调整
\end{align}
where $\Delta _{nb}^i\left( {{n_1},{n_2}} \right) = J_{\rm{nb}}^i\left( {{n_1}} \right){\left[ {J_{\rm{nb}}^i\left( {{n_2}} \right)} \right]^*}$, $\Delta n=n_2-n_1$, $(a)$ follows from that $\mathbb{E}\left\{ {\Delta _{nb}^i\left( {{n_1},{n_2}} \right)} \right\}$ depends solely on difference between $n_1$ and $n_2$.
When $i=1$, 
\begin{align}\label{eq:E_S1_1}
\vspace*{-2pt} %留空白，可自己调整
\mathbb{E}\left\{ {{S_1}} \right\} {=} N &{\cdot} \mathbb{E}\left\{ {\sum\limits_{{k_1} =  - \infty }^\infty  {{{\left| {z\left( { - {k_1}} \right)} \right|}^2}{{\left| {h\left( {{B_{\rm{i}}},{k_1}} \right)} \right|}^2}} } \right\} + \nonumber\\
N &{\cdot} \mathbb{E}\left\{ {\sum\limits_{{k_1} =  - \infty }^\infty  {\sum\limits_{{k_2} \ne {k_1}} {z\left( { - {k_1}} \right){z^*}\left( { - {k_2}} \right){\Delta _h}\left( {{k_1},{k_2}} \right)} } } \right\}\nonumber\\
{=} N  &{\cdot} 1 + N \cdot 0=N
\text{,}
\vspace*{-2pt} %留空白，可自己调整
\end{align}
where ${\Delta _h}\left( {{k_1},{k_2}} \right) = h\left( {{B_{\rm{i}}},{k_1}} \right){h^*}\left( {{B_{\rm{i}}},{k_2}} \right)$. While $i=2$, we have
\begin{equation}
\vspace*{-2pt} %留空白，可自己调整
{
	\label{eq:E_S1_2}
	\mathbb{E}\left\{ {{S_1}} \right\} = N \cdot \mathbb{E}\left\{ {{{\left| {a\left( {\left\lceil {\frac{{{n_1}}}{{{R_{\rm{u}}}}}} \right\rceil } \right)} \right|}^2}} \right\} = N\text{.} }
\vspace*{-2pt} %留空白，可自己调整
\end{equation}
Thus we have
\begin{equation}
\vspace*{-2pt} %留空白，可自己调整
{
	\label{eq:V_NBI}
	\mathbb{V}\left\{ {J_{\rm{nb}}^{\rm{A}}\left( m \right)} \right\} = \frac{{{P_{\rm{i}}}}}{N}\left( N+0 \right) = {P_{\rm{i}}}\text{.} }
\vspace*{-2pt} %留空白，可自己调整
\end{equation}

Proposition 4 is proved.

\section{Proof of Proposition 5}\label{proof_Corollary5}
Given that AFDM enables full diversity gain in doubly selective channels \cite{FWC,AFDM_TWC}, we formulate the signal after despreading as
\begin{equation}
\vspace*{-2pt} %留空白，可自己调整
{
	\label{eq:receive}
	\hat c\left[ n \right] = {G_{\rm{c}}}c\left( n \right) + I\left( n \right)  \text{,} }
\vspace*{-2pt} %留空白，可自己调整
\end{equation}
where ${G_{\rm{c}}} = {N_{\rm{d}}}\sqrt {{{{h_L}{P_{\rm{s}}}} \mathord{\left/{\vphantom {{{h_L}{P_{\rm{s}}}} {{{\log }_2}\left( {{N_{\rm{m}}}} \right)}}} \right.\kern-\nulldelimiterspace} {{{\log }_2}\left( {{N_{\rm{m}}}} \right)}}}$ indicates effective signal gain obtained after exploiting both the full diversity gain and the spreading gain, ${h_L} = \sum\limits_{i = 0}^{L - 1} {{{\left| {{h_i}} \right|}^2}}$, and 
\begin{equation}
\vspace*{-2pt} %留空白，可自己调整
{
	\label{eq:IN_all}
	I\left( n \right) = \sum\limits_{m = \frac{{n{N_{\rm{d}}}}}{{{{\log }_2}\left( {{N_{\rm{m}}}} \right)}}}^{\frac{{\left( {n + 1} \right){N_{\rm{d}}}}}{{{{\log }_2}\left( {{N_{\rm{m}}}} \right)}} - 1} {{d_{\rm{s}}}\left( {{m^{'}}} \right)\left[ {J\left( m \right) + w\left( m \right)} \right]}    \text{,} }
\vspace*{-2pt} %留空白，可自己调整
\end{equation}
where ${{d_{\rm{s}}}\left( {{m}} \right)}$ is spreading sequence $d\left( {{n}} \right)$ under constellation mapping, ${m^{'}}={\left\langle m \right\rangle _{{N_{\rm{d}}}/{{\log}_2}\left( {{N_{\rm{m}}}} \right)}}$, $ w\left(m\right)\sim \mathcal {CN}\left( {0\text{,} P_{\rm{n}}} \right)$ is additive Gaussian noise in the DAFT domain, ${J}\left( m \right)$ denotes malicious interference in the DAFT domain, which can take the form of ${J_{\rm{st}}^{\rm{A}}}\left( m \right)$, ${J_{\rm{sw}}^{\rm{A}}}\left( m \right)$, ${J_{\rm{bb}}^{\rm{A}}}\left( m \right)$, or ${J_{\rm{nb}}^{\rm{A}}}\left( m \right)$. 
We derive BER under two cases, i.e., stationary interference impact and non-stationary interference impact.

\textit{i) Stationary interference impact in the DAFT domain:} Leveraging an impulse-like autocorrelation function of spreading sequence $d\left(n\right)$, each value of the spreading sequence could be considered as statistically independent. Based on the central limit theorem, Proposition 1, Proposition 2, Proposition 3 and Proposition 4, we have
\begin{equation}
\vspace*{-2pt} %留空白，可自己调整
{
	\label{eq:IN_GA}
	I\left( n \right) \sim \mathcal {CN}\left( {0\text{,} N_{\rm{d}}\left(P_{\rm{i}}+P_{\rm{n}}\right)} \right) \text{.} }
\vspace*{-2pt} %留空白，可自己调整
\end{equation}
Let $\xi  = {{{N_{\rm{d}}}{P_{\rm{s}}}} \mathord{\left/{\vphantom {{{N_{\rm{d}}}{P_{\rm{s}}}} {\left[ {{{\log }_2}\left( {{N_{\rm{m}}}} \right)\left( {{P_{\rm{i}}} + {P_{\rm{n}}}} \right)} \right]}}} \right.\kern-\nulldelimiterspace} {\left[ {{{\log }_2}\left( {{N_{\rm{m}}}} \right)\left( {{P_{\rm{i}}} + {P_{\rm{n}}}} \right)} \right]}}$, $P_{\rm{e}}$ given $h_L$ can be expressed as $Q\left( {\sqrt {2{h_L}\xi } } \right)$. And $h_L$ follows a Gamma distribution \cite{AFDM_TWC}, i.e., 
\begin{equation}
\vspace*{-2pt} %留空白，可自己调整
{
	\label{eq:h_L}
	{h_L} \sim \Gamma\left( {L,\frac{1}{L}} \right) \text{.} }
\vspace*{-2pt} %留空白，可自己调整
\end{equation}
Following \cite[Eq. 3-37]{FWC}, $P_{\rm{e,s}}$ could be given by
\begin{align}\label{eq:Pe}
\vspace*{-2pt} %留空白，可自己调整
{P_{\rm{e,s}}} =\frac{1}{2} - \frac{1}{2}\sum\limits_{i = 0}^{L - 1} { \binom{2i}{i}} {\left( {\frac{{{\gamma _{\rm{in}}}}}{{4{N_{\rm{d}}}}}} \right)^i}{\left( {\frac{{{N_{\rm{d}}}}}{{{N_{\rm{d}}} + {\gamma _{\rm{in}}}}}} \right)^{i + 0.5}}
\text{,}
\vspace*{-2pt} %留空白，可自己调整
\end{align}
where ${\gamma _{\rm{in}}} = {{L{{\log }_2}\left( {{N_{\rm{m}}}} \right)\left( {{P_{\rm{n}}} + {P_{\rm{i}}}} \right)} \mathord{\left/{\vphantom {{L{{\log }_2}\left( {{N_{\rm{m}}}} \right)\left( {{P_{\rm{n}}} + {P_{\rm{i}}}} \right)} {{P_{\rm{s}}}}}} \right.\kern-\nulldelimiterspace} {{P_{\rm{s}}}}}$.

\textit{ii) Non-stationary interference impact in the DAFT domain:} If the spreading sequence associated with $\hat c\left[ n \right]$ contains no interfered chips, the interference term $I\left(n\right)$ is dominated by $w\left(m\right)$. Conversely, when the spreading sequence corresponding to $\hat c\left[ n \right]$ includes interfered chips, the interference term $I\left(n\right)$ follows a Gaussian distribution, simliar to (\ref{eq:IN_GA}). Thus, we have
\begin{align}\label{eq:IN_SW}
\vspace*{-2pt} %留空白，可自己调整
\begin {cases}
I\left( n \right) \sim \mathcal {CN}\left( {0\text{,} {N_{\rm{d}}}P_{\rm{n}}} \right) \quad &{\alpha  \notin {\mathbb{T}}\left(n\right)} \\
I\left( n \right) \sim \mathcal {CN}\left( {0\text{,} \frac{{{N_{\rm{d}}}^2{P_{\rm{n}}} + {N}{P_{\rm{i}}}}}{{{N_{\rm{d}}}}}} \right)  \quad &{\alpha   \in  {\mathbb{T}}\left(n\right)} \\
\end{cases},
%\text{,}
\vspace*{-2pt} %留空白，可自己调整
\end{align}
where ${\mathbb{T}\left(n\right)}=\left[ {{{\left\langle {n{N_r}} \right\rangle }_N},} \right.\left. {{{\left\langle {\left( {n + 1} \right){N_r}} \right\rangle }_N}} \right)$ and ${N_r} = {{{N_{\rm{d}}}} \mathord{\left/
		{\vphantom {{{N_{\rm{d}}}} {{{\log }_2}\left( {{N_{\rm{m}}}} \right)}}} \right.
		\kern-\nulldelimiterspace} {{{\log }_2}\left( {{N_{\rm{m}}}} \right)}}$. Similar to (\ref{eq:Pe}), $P_{\rm{e}}$ could be given by
\begin{align}\label{eq:Pen}
\vspace*{-2pt} %留空白，可自己调整
{{P_{\rm{e,n}}}\left( n \right)}
\begin {cases}
\frac{1}{2} - \frac{1}{2}\Psi_{1} \left( {{N_{\rm{d}}}} \right) \quad &{\alpha  \notin {\mathbb{T}}\left(n\right)} \\
\frac{1}{2} - \frac{1}{2}\Psi_{2} \left( {{N_{\rm{d}}}} \right)   \quad &{\alpha   \in  {\mathbb{T}}\left(n\right)} \\
\end{cases},
%\text{,}
\vspace*{-2pt} %留空白，可自己调整
\end{align}	
where
${\gamma _{\rm{n}}} = {{L{{\log }_2}\left( {{N_{\rm{m}}}} \right){P_{\rm{n}}}} \mathord{\left/
		{\vphantom {{L{{\log }_2}\left( {{N_{\rm{m}}}} \right){P_{\rm{n}}}} {{P_{\rm{s}}}}}} \right.
		\kern-\nulldelimiterspace} {{P_{\rm{s}}}}}$, ${\gamma _{\rm{m}}} = {\gamma _{\rm{i}}} + {\gamma _{\rm{n}}}{N_{\rm{d}}}^2$ and
${\gamma _{\rm{i}}} = {{{N}L{{\log }_2}\left( {{N_{\rm{m}}}} \right){P_{\rm{i}}}} \mathord{\left/
		{\vphantom {{{N}L{{\log }_2}\left( {{N_{\rm{m}}}} \right){P_{\rm{i}}}} {{P_{\rm{s}}}}}} \right.
		\kern-\nulldelimiterspace} {{P_{\rm{s}}}}}$.
From the definition of ${\mathbb{T}\left(n\right)}$, ${P_{\rm{e,n}}} = \frac{1}{2} - \frac{1}{2}\left[R{\Psi _2}\left( {{N_{\rm{d}}}} \right) + \left( {1 - R} \right){\Psi _1}\left( {{N_{\rm{d}}}} \right)\right]$ and $R={{{{N_{\rm{d}}}} \mathord{\left/
			{\vphantom {{{N_{\rm{d}}}} {\left[ {N{{\log }_2}\left( {{N_{\rm{m}}}} \right)} \right]}}} \right.
			\kern-\nulldelimiterspace} {\left[ {N{{\log }_2}\left( {{N_{\rm{m}}}} \right)} \right]}}}$.

Proposition 5 is proved.

\small
\bibliographystyle{IEEEbib}
\bibliography{IEEEabrv,IEEE_JRCJ_ref}

\end{document}